\documentclass{myElsart}

\usepackage[figuresright]{rotating}
\usepackage[colorlinks=true,citecolor=blue,linkcolor=blue]{hyperref}
\usepackage{amsmath}
\usepackage{amssymb}
\usepackage[bb=boondox]{mathalfa} 
\usepackage{bbm}
\usepackage{mathtools}
\usepackage{booktabs}
\usepackage{theorem}
\usepackage{algpseudocode}
\usepackage{algorithm}
\usepackage[]{natbib}
\usepackage{tikz}
\usetikzlibrary{arrows}
\usepackage{enumitem}

\newtheorem{theorem}{Theorem}[section]
\newtheorem{lemma}[theorem]{Lemma}

\newenvironment{proof}[1][Proof]{\vspace{1em}\begin{trivlist}
\item[\hskip \labelsep {\bfseries #1}]}{\hfill $\Box$\end{trivlist}\vspace{1em}}

\begin{document}

\begin{frontmatter}

\title{Blume-Capel model: Estimation of a three stable state network for $-\bf 1$, $\bf 0$ and $\bf +1$ data}
\author{}\\
\author{Lourens Waldorp$^{1,\star}$ ({\tt waldorp@uva.nl})} 
\author{Jonas Dalege$^1$ ({\tt j.dalege@uva.nl})}
\author{Maarten Marsman$^1$ ({\tt m.marsman@uva.nl})}
\author{Adam Finnemann$^1$ ({\tt a.t.k.finnemann@uva.nl})}
\author{Irene Ferri$^2$ ({\tt irene.ferry@ub.edu})}
\author{Han L. J. van der Maas$^1$ ({\tt h.l.j.vdmaas@uva.nl})}
\vspace{2em}

\author{$^1$ Department of Psychological Methods, University of Amsterdam, Netherlands}
\author{$^2$ Northeastern University - Portland Campus, Portland, United States}
\vspace{2em}





\begin{abstract}
An extension of the Ising model is proposed as a viable alternative for data with values $-1$, $0$ and $+1$ in the inverse problem, i.e., estimation of the parameters. This model is called the Blume-Capel (BC) model, adapted from physics for small networks. The advantage of the BC model is not only the fact that it is possible to have a neutral (centrist) position on the response scale, but also that this model allows for three stable states. We illustrate magnetisation properties of the BC model using simulations and mean field results. For estimation of the BC parameters, we show that the BC model is part of the exponential family of distributions and show that the model is identified, except for the (inverse) temperature. We then show that combining pseudo-likelihood with lasso yields accurate parameter recovery for the BC model, even in small networks. Moreover, confidence intervals with good coverage properties can be obtained using the desparsified lasso together with sandwich and shrinkage techniques. We apply the methods to data obtained from the online platform \textit{Stemwijzer}, intended to aid people in deciding for whom to vote. 
\end{abstract}

\end{frontmatter}

\endNoHyper

\section{Introduction}\noindent
The Blume-Capel (BC) model \citep{Blume:1966,Capel:1966} is an extension of the Ising model where, next to the values $-1$ and $+1$, the value 0 is introduced to the possible states of each random variable (node) in the network. The reason that the value 0 is introduced is to emphasise a central or centrist position. For example, in sociophysics $-1$ is identified with leftist, $+1$ with rightist and $0$ with centrist positions on a political scale \citep{Fernandez:2016, Ferri:2022}. Hence, introducing the value 0 into the model has significant meaning. Such meaning can also be found in psychology. In psychology the response 0 could refer to ``don't know'' or a central (neutral) position with respect to a statement like ``Eating meat is bad for the environment.''  \citep{Maas:2026}.
Here we introduce some properties of the BC model in terms of probabilities and dynamics and obtain a pseudo-likelihood procedure (using the lasso) to obtain estimates and confidence intervals of the parameters of the BC-model, i.e., the inverse problem \citep[see, e.g.,][for the Ising model]{Albert:2014,Epskamp:2015aa,Marsman:2018}.

The BC model shows richer and different behaviours than the Ising model. First, in terms of stationary (equilibrium) distributions, the BC model has three possible minima (at low temperature); in the phase space of the BC model there are parameter settings such that three different states are possible \citep{Kaufman:1990}. This is a significant contribution to the existing possibilities of the Ising model, which only allows two stable states. 
The third stable position, neutrality, signifies that when the temperature is low enough and opinions tend to align, three different groups in political views or attitudes toward the environment, for example, could be obtained. 
Second, the BC model is the simplest model that allows a sudden change (first-order transition) without an external field. In the BC model, increasing the parameter that controls the prevalence of neutral states ($0$s) induces a first-order phase transition: the system abruptly shifts from an Ising-like regime with states $-1$ and $+1$ to one where the $0$ state dominates (see Section \ref{sec:blume-capel-model}). This transition occurs without the need for an external field \citep{Fernandez:2016}, in contrast to the Ising model \citep{Bouabci:2000}. And third, in the regime where a first-order transition occurs, the BC model exhibits hysteresis: the transition between states takes place at different values of the control parameter depending on whether it is increased or decreased \citep{Akinci:2016}. These properties of the BC model make it an attractive alternative to the Ising model.

The Blume-Capel model was introduced as a model for the magnetic material uranium dioxide, independently by  \cite{Blume:1966} and \citet[][Hans Capel, from both Leiden University and the University of Amsterdam]{Capel:1966}. Several years later an extension of the BC model was applied to mixed forms of helium  \citep[Blume-Emery-Griffiths model][]{Blume:1971}. But the BC model turned out to be a simpler model capturing many interesting effects and allowing for more rigorous study (like the Ising model) than other extensions \citep[see, e.g., ][]{Bouabci:2000,Graham:2006}.

Several alternatives to the Ising model also allow for multiple states. For instance, the Potts model \citep{Plischke:1994,Grimmett:2022} with three states could be used, but changes energy only if the states are equal. Also, a multinomial logistic graphical model is an alternative \citep{Yang:2012,Suggala:2017}, but does not have the neutrality parameter, which we show in Section \ref{sec:application-to-data} has a clear and relevant connection to the proportion of zeros.

We first introduce the BC model in Section \ref{sec:blume-capel-model}, describing equilibrium and dynamic (mean field) properties. We also characterise the BC model in terms of the exponential family distributions, showing that the model is identified except for one parameter (inverse temperature). Then, in Section \ref{sec:pseudo-likelihood} we introduce the pseudo-likelihood method for the BC model. This is a necessary step, as the normalisation constant of the BC model, like with the Ising model, is computationally infeasible. We also introduce the lasso estimation and show that standard errors can be obtained for the pseudo-likelihood (and lasso) by applying the sandwich estimate. We prove (in Appendix \ref{app:dLshrink}) that our proposed method leads to accurate confidence intervals given the appropriate conditions. We give a numerical illustration in Section \ref{sec:numerical-illustrations}, where we provide some evidence that the estimation is quite reasonable. Then, finally, we apply the model to data on voting behaviour in Section \ref{sec:application-to-data}.

\section{Blume-Capel model}\label{sec:blume-capel-model}\noindent
We discuss the Blume-Capel (BC) model that allows for each variable in the network to obtain the three different values $-1$, $0$ and $+1$. As reasoned above, it is important for the applications we have in mind (e.g., attitudes, or depression) that the values are symmetric about 0, so that 0 takes the central position in the set of options. 

We assume a finite graph $\mathcal{G}$ that consists of nodes in $V$ and edges in $E$. We write nodes as $1,2,\ldots, m$ and edges as $(s,t)\in E$. The graph has an arbitrary topology fixed in $E$ and one of the aims is to determine the set of edges in $E$. Associated with each node is a random variable $X_s$ for $s\in V$; we refer to the vector of all random variables in the graph as $X=(X_1,X_2,\ldots,X_m)$ and their observed values $x=(x_1,x_2,\ldots,x_m)$. The joint distribution of all random variables that describes the graph $\mathcal{G}$ in terms of the edges is called a graphical model \citep{Maathuis:2018}. The joint distribution is here given by the BC model. 

The BC model is characterised by the Hamiltonian. Similar to the Ising model \citep{Waldorp:2019} the BC model has threshold parameters $\tau_s$ for all variables $X_s$, interaction (connection) parameters $\sigma_{st}$ for all edges $(s,t)$ in $E$. The Hamiltonian of the BC model is
\begin{align}\label{eq:hamiltonian-bc}
\mathcal{H}(x) =  -\sum_{s\in V} \tau_s x_s - \sum_{(s,t)\in E} \sigma_{st} x_s x_t + \alpha^2 \sum_{s\in V} x_s^2 
\end{align}
This Hamiltonian differs from the Ising model in that there is an additional term $\alpha^2 \sum_{s\in V} x_s^2$. This term co-determines the probability of $x_s$ being zero or non-zero. The parameter $\alpha^2$ "punishes" for the number of non-zero values in $x$; high values of $\alpha^2$ lead to more 0s in $x$ when generating values (see Section \ref{sec:equilibrium-dynamics}).

The corresponding distribution of the BC model can be used as a representation of an equilibrium distribution, at the end of some process involving change (e.g., learning or forming an opinion). As with the Ising model, when a system with interactions governed by the BC model is placed in thermal equilibrium at temperature $T$, the configuration probabilities follow a Boltzmann distribution with inverse temperature $\beta=1/T$. We collect all parameters $\tau_s$, $\sigma_{st}$ and $\alpha^2$ in the vector $\theta$ of dimension $m(m+1)/2+1$. Then the distribution of the BC model is
\begin{align}\label{eq:joint-distribution}
p_\theta(x) = \frac{1}{Z_\theta}\exp\left( \beta\sum_{s\in V} \tau_s x_s +\beta \sum_{(s,t)\in E} \sigma_{st} x_s x_t - \beta\alpha^2 \sum_{s\in V} x_s^2 \right)
\end{align}
where the normalising constant (partition function) is 
\begin{align}\label{eq:normalising-constant}
Z_\theta = \sum_{x\in \{-1,0,+1\}^V} \exp\left( \beta\sum_{s\in V} \tau_s x_s + \beta\sum_{(s,t)\in E} \sigma_{st} x_s x_t - \beta\alpha^2 \sum_{s\in V} x_s^2\right)
\end{align}

To characterise the BC model we illustrate by simulation and the mean field approximation some equilibrium and dynamical properties. 

\subsection{Equilibrium and dynamics characterisation}\label{sec:equilibrium-dynamics}\noindent
Using a Gibbs sampler for the BC model (see Appendix \ref{app:gibbs-sampler}) we obtain the probabilities of $-1$, $0$, and $+1$ values. In Figure \ref{fig:bc-alpha}(a) the proportion of values $-1$, $0$ and $+1$ (magnetisation) for different values of $\alpha^2$ is uniform because in (\ref{eq:joint-distribution}) the Boltzmann distribution is fat when $\beta=0$. However, when $\beta=2$ thermal agitation is suppressed and the system tends to settle into energetically favourable configurations, we see that the parameter $\alpha^2$ controls the proportion of 0s (see Figure \ref{fig:bc-alpha}(b)). 
%
\begin{figure}\centering
\begin{tabular}{@{\hspace{-1em}} c @{\hspace{2em}} c @{\hspace{1em}} c}
		&\pgfimage[width=.45\textwidth]{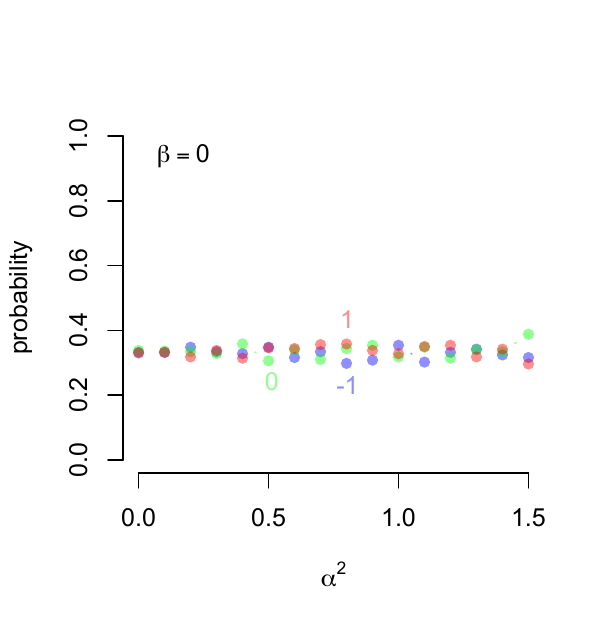}		&\pgfimage[width=.45\textwidth]{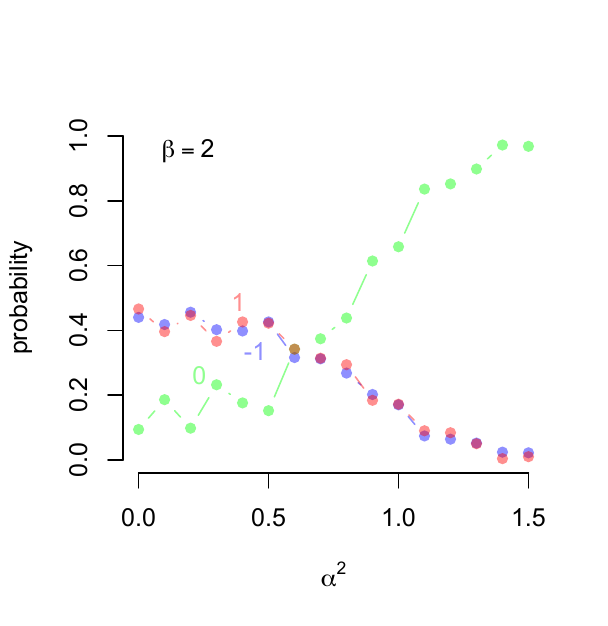}\\
	&	(a)		&(b)
\end{tabular}
\caption{Average proportions of a 5 node BC model (in equilibrium) observed $100$ times of $-1$ (blue), $0$ (green) and $1$ (red). Parameter settings are $\tau=0$  and $\beta=0$ in (a) and $\beta=2$ in (b) as a function of $\alpha^2$, controling the proportion of 0s when $\beta$ is large enough (temperature is low). }
\label{fig:bc-alpha}
\end{figure}
The transition from zero probability of 0s ($\alpha^2=0$) to non-zero probability of 0s ($\alpha^2>0$) is known to be a sudden transition, a first-order transition.
The parameter $\alpha^2$ influences this transition by shifting the critical temperature at which it occurs: increasing $\alpha^2$ favours the neutral (0) state and moves the transition point accordingly \citep{Bouabci:2000}. Right below the critical temperature, the system settles into the $\pm1$-dominated phase, while right above it, the 0-dominated phase prevails. This can be seen in Figure \ref{fig:bc-alpha}(b), where the three possible values for the magnetization exist for a given range of $\alpha^2$ (in a small network the transition is less sudden). This is contrasted with the well-known second-order (continuous) transition for temperature $\beta$ that is also present in the Ising model \citep{Kindermann:1980,Grimmett:2006b}. In the BC model, first- and second-order transition lines meet at the tricritical points, where the three degenerate free-energy minima that coexist along the first-order line coalesce into a single flat minimum, as both the second and fourth derivatives of the free energy with respect to the order parameter vanish \citep{Kaufman:1990,Plascak:1993}.

To characterise the BC model in terms of its dynamics, we use mean field theory \citep[e.g.,][]{Plischke:1994,Barrat:2008}. It is then possible to show (by large deviation theory) that the mean field results are reasonably close to what can be expected from the original model \citep[e.g.,][]{Wainwright:2008,Waldorp:2020}. We then obtain (see Appendix \ref{app:mean-field} for details)
\begin{align}\label{eq:mean-field}
\mu 	&= \frac{2\sinh(\beta( \tau + \mu d\sigma))\exp(-\beta\alpha^2)  }
		{ 1 + 2\cosh( \beta(\tau + \mu d\sigma))\exp( -\beta\alpha^2 ) }
\end{align}
The mean field approximation gives the "magnetisation", the expectation of any of the variables in the graph as a function of the mean (average) magnetisation in the graph $\mu$, the threshold $\tau$, the interaction parameter $\sigma$, the parameter $\alpha^2$ controlling the proportion of 0s, and $d$ the average number of connections in the graph. 

Figure \ref{fig:dynamics} shows for particular parameter settings ($\beta=2$, $\tau=0$, $\sigma=1$, $\alpha^2=2$, and $d=5$) where the magnetisation of the BC model (approximated using mean field theory) will end up. In Figure \ref{fig:dynamics}(a) the Gibbs free energy is shown (see equation (\ref{eq:free-energy})), illustrating the three stable fixed points (blue circles) and two unstable fixed points (red circles). In Figure \ref{fig:dynamics}(b) we see the same fixed points of the mean field approximation  (\ref{eq:mean-field}), again with $\alpha^2=2$, at the points where the diagonal line intersects the curve. Figure \ref{fig:dynamics}(a) and (b) show attracting fixed points (blue circles) where the magnetisation can end up, and repelling fixed points (red circles) which the mean field will move away from. In contrast to the Ising model, there is now a third stable fixed point at 0. 
This is confirmed by Figure \ref{fig:dynamics}(c), showing a bifurcation diagram. The bifurcation diagram shows for each parameter $\alpha^2$ on the $x$-axis the value of $\mu^n(\mu_0)$, the iterated mean field map after $n$ iterations (we used $n=100$ here) starting at the random     initial value $\mu_0$ \citep{Hirsch:2004,Holmgren:1996}. 
A bifurcation can be seen in Figure \ref{fig:dynamics}(c) at $\alpha^2\approx 2.5$, where the three fixed points change to a single fixed point, such that $\mu=0$, corresponding to the case where all variables are 0.  
\begin{figure}\centering
\begin{tabular}{@{\hspace{-0.5em}} c @{\hspace{-0.7em}} c @{\hspace{-0.7em}} c @{\hspace{-0.7em}} c}
		&\pgfimage[width=.37\textwidth]{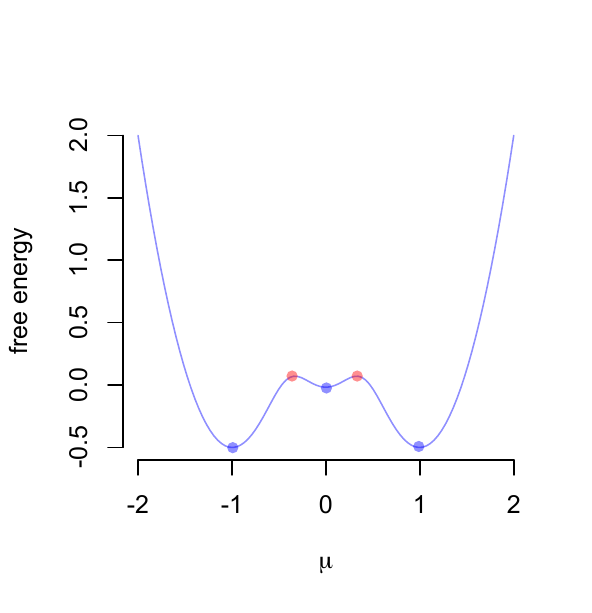}		
        &\pgfimage[width=.37\textwidth]{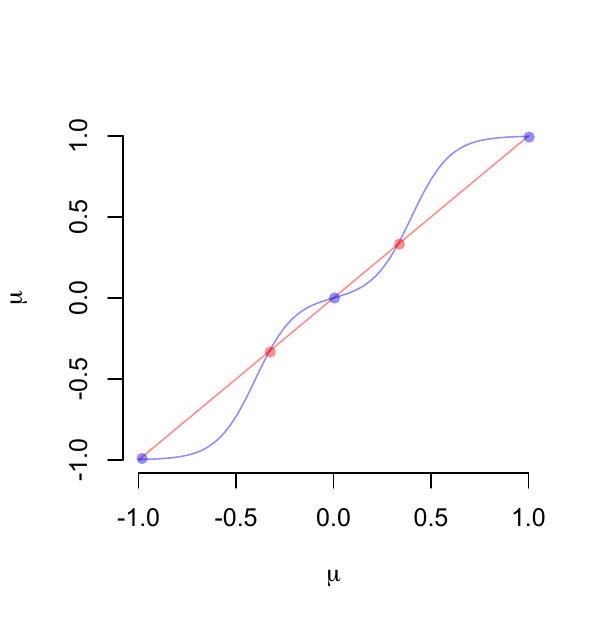}
        &\pgfimage[width=.37\textwidth]{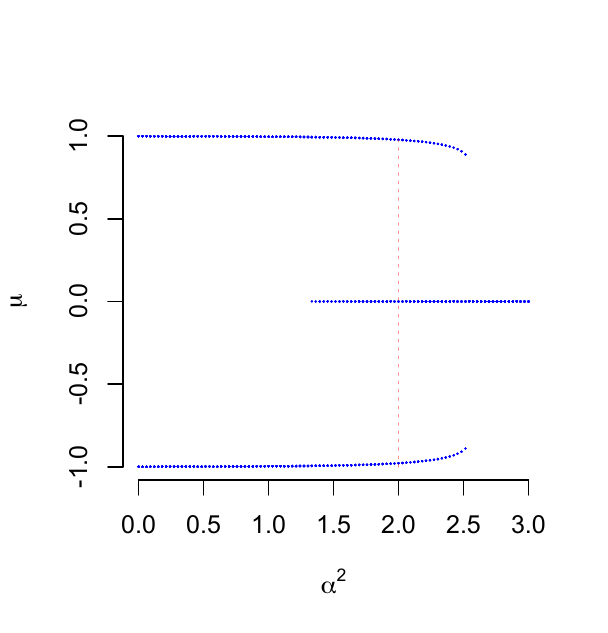}\\
	&	(a)		&(b)    &(c)
\end{tabular}
\caption{In (a) is the free energy (equation (\ref{eq:free-energy})), showing three stable fixed points. The settings are $\beta=2$, $\alpha^2=2$, $\tau=0$, $\sigma=1$ and $d=5$. The red circles are repelling fixed points and the blue circles are attracting fixed points. In (b) is the mean field approximation $\mu$ (see equation (\ref{eq:mean-field})) for the same settings; the intersection with the diagonal (red) line indicates fixed points (circles), with the same interpretations as in (a). In (c) is a bifurcation diagram as a function of $\alpha^2$ (and otherwise the same parameter settings as in (a)), showing points where after 500 iterations of the mean field map $\mu$ the system ends up. The vertical dotted line represents $\alpha^2=2$, corresponding to (a). }
\label{fig:dynamics}
\end{figure}
%

\subsection{Exponential family characterisation}\label{sec:exponential-family}\noindent
The BC model is a full (finite normalising constant) and minimal (no redundancies in the sufficient statistics) exponential model of the exponential family distributions, given that the parameter $\beta$ is absorbed by the threshold parameters $\tau_s$ and the interaction parameters $\sigma_{st}$ (see Appendix \ref{app:minimal-sufficient-statistic} for a proof). For the external field $\tau_s$, the interactions $\sigma_{st}$ and the non-zero cost $\alpha^2$ parameters, minimality implies that the model is identifiable, i.e., the mapping $\theta \mapsto p_\theta(x)$ is injective \citep{Barndorff-Nielsen:1994,Wainwright:2008}, and this is the case we discuss here. 

The fact that we are dealing with a minimal sufficient exponential family distribution allows for a bijective map between the space of the means $\mu_s=\mathbb{E}(\phi_s(X))$, where $\phi_s(x)$ is either $x_s$, $x_s^2$ or $x_s x_t$, and the natural parameter space $\Omega$. The map is given by the derivative of the log partition function $\nabla \log Z_\theta: \Omega\to \mathbb{M}$, where $\Omega$ is the natural parameter space such that $\theta\in \Omega$ and $\mathbb{M}$ is the space of means of the sufficient statistics $\mu_s=\mathbb{E}(\phi_s(X))\in \mathbb{M}$ \citep[][Proposition 3.2, see Appendix \ref{app:properties-minimal-exponential-family} for more details]{Wainwright:2008}. This mapping is bijective for points in the interior of the image $\mathbb{M}$, where the gradient of the log-partition function is well-defined and strictly convex. On the boundary of $\mathbb{M}$ this bijection fails, and the maximum likelihood estimator may not exist or may diverge. The inverse map yields unique maximum likelihood estimates $\hat{\theta}\in \Omega$ within the full parameter space $\mathbb{R}^d$, where $d=|V|+|E|+1$ for the BC model; for $|V|=m$ nodes we have $m(m+1)/2+1$ parameters. 

We give a small example to illustrate some properties of the BC model and the relation between the sufficient statistics and the means of the BC model. The simplest non-trivial case is a graph with two nodes and one edge. Then $x\in \{-1,0,+1\}^2$ and so the sufficient statistic is $\phi(x)=(x_1,x_2, x_1x_2,x_1^2+x_2^2)$. It turns out that for a finite space like $\{-1,0,+1\}^2$, the means $\mu_s=\mathbb{E}(\phi_s(X))$ are obtained by considering the space inside of all sufficient statistics. The set of possible outcomes of $\phi(x_1,x_2)$ with $(x_1,x_2)\in \{-1,0,+1\}^2$ is given by
\begin{align*}
\Phi_2= \{  x_1,x_2,x_1x_2, x_1^2+x_2^2\mid (x_1,x_2)\in \{-1,0,+1\}^2  \}
\end{align*}
For this example we have $3^2$ possibilities for two nodes, and so can write out the set of points in $\Phi_2$ 
\begin{align*}
\begin{matrix*}[l]
\phi(-1,-1)=(-1,-1,+1,2)		&\quad\phi(-1,0)=(-1,0,0,+1)\\
\phi(-1,+1)=(-1,+1,-1,2)		&\quad\phi(0,-1)=(0,-1,0,+1)\\
\phi(+1,+1)=(+1,+1,+1,2)		&\quad\phi(0,+1)=(0,+1,0,+1)\\
\phi(+1,-1)=(+1,-1,-1,2)		&\quad\phi(+1,0)=(+1,0,0,+1)\\
\phi(0,0)=(0,0,0,0)
\end{matrix*}
\end{align*}
The space with all means $\mu_s=\mathbb{E}(X_s)$ and $\mu_{st}=\mathbb{E}(X_sX_t)$ for all $s,t\in \{1,2\}$ is the convex hull of these points \citep[i.e., the marginal polytope,][]{Wainwright:2008}. For our small example this is
\begin{align*}
\mathbb{M}=\{\mu\in \mathbb{R}^4\mid \mu_s=\mathbb{E}(X_s), \mu_{st}=\mathbb{E}(X_sX_t), s,t\in \{1,2\}\}
\end{align*}
Figure \ref{fig:marginal-polytope}(a) and \ref{fig:marginal-polytope}(b) show the four-dimensional space split up into two separate figures (these two figures are connected with each other at the vertices). In Figure \ref{fig:marginal-polytope}(a) is the polytope for the case where $x_1$ and $x_2$ are non-zero such that $x_1^2+x_2^2=2$ (this is the first column of points in $\Phi_2$). In \ref{fig:marginal-polytope}(a) this is represented as the triangular shape. In Figure \ref{fig:marginal-polytope}(b) is the polytope for the case where one of $x_1$ or $x_2$ is 0 such that $x_1^2+x_2^2=1$ (the second column of points in $\Phi_2$). Then there is a final point in $\Phi_2$ which is $(0,0,0,0)$.  
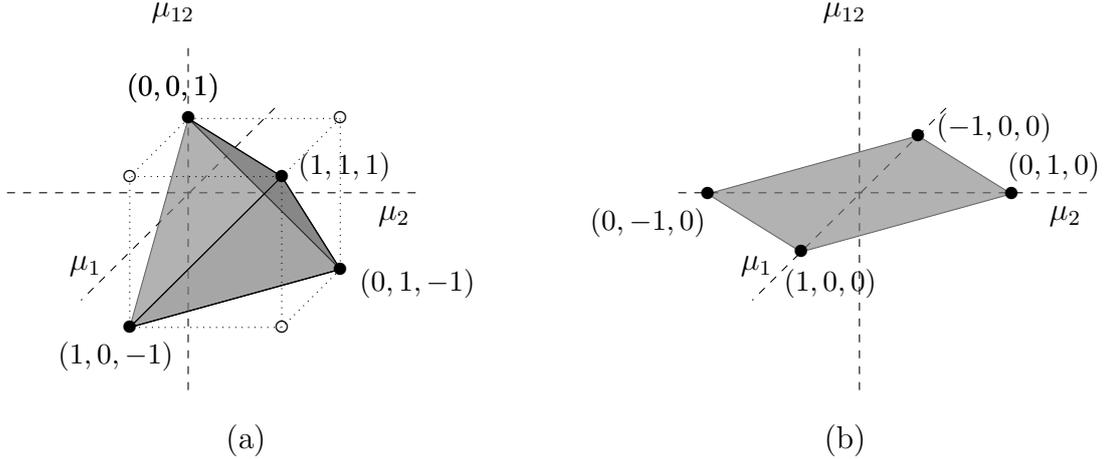
\begin{figure}
\begin{tabular}{c @{\hspace{3em}} c }
\begin{tikzpicture}[scale=2,>=stealth]   
    \draw[dashed] (-1,0.7,0.5) -- (1.7,0.7,0.5); 	
    \draw[dashed] (0.5,-0.3,1.3) -- (0.5,2,1.3);	
    \draw[dashed] (0.3,0.8,-0.7) -- (0.45,0.95,3);	

    \draw[fill=gray,opacity=0.6] (0,0,1) -- (0,1,0) -- (1,1,1) -- cycle;	
    \draw[fill=gray,opacity=0.7] (1,0,0) -- (0,0,1) -- (1,1,1) -- cycle;	
    \draw[fill=gray,opacity=0.8] (1,0,0) -- (0,1,0) -- (1,1,1) -- cycle;	
    \draw[dotted] (1,0,0) -- (1,0,1) -- (0,0,1);	
    \draw[dotted] (0,1,0) -- (0,1,1) -- (0,0,1);	
    \draw[dotted] (1,0,1) -- (1,1,1) -- (0,1,1);	
    \draw[dotted] (0,1,0) -- (1,1,0) -- (1,0,0);	
    \draw[dotted] (1,1,0) -- (1,1,1);         
    \draw[-] (1,1,1) -- (0,1,0);	
    \draw[-] (1,1,1) -- (1,0,0);	
    \draw[-] (0,0,1) -- (1,0,0);	
    \draw[-] (0,0,1) -- (1,0,0);	
    \draw[-] (0,0,1) -- (1,1,1);	
    \draw (1.7,0.7,.4+0.5) node{$\mu_{2}$};		
    \draw (-.2+0.1,0.6,1.5) node{$\mu_{1}$};	
    \draw (-.1,1.7,0) node{$\mu_{12}$};	
    \draw (1,1,1) node{$\bullet$};		
    \draw (0,1,1) node{$\circ$};		
    \draw (0,0,1) node{$\bullet$}; 	
    \draw (0,1,0) node{$\bullet$};		
    \draw (1,0,0) node{$\bullet$};		
    \draw (1,0,1) node{$\circ$};		
    \draw (1,1,0) node{$\circ$};		
    \draw (0.1,0,1.5) node{\small $(1,0,-1)$};		
    \draw (1.7,0.1,0.5) node{\small $(0,1,-1)$};		
    \draw (-0.1,1.2,0) node{\small $(0,0,1)$};		
    \draw (-0.1,1.2,0) node{\small $(0,0,1)$};		
    \draw (-0.05,-0.4,-2.8) node{\small $(1,1,1)$};	
  \end{tikzpicture}  

&

\begin{tikzpicture}[scale=2,>=stealth]   
    \draw[dashed] (-1,0.7,0.5) -- (1.7,0.7,0.5); 	
    \draw[dashed] (0.5,-0.3,1.3) -- (0.5,2,1.3);	
    \draw[dashed] (0.3,0.8,-0.7) -- (0.45,0.95,3);	
    \draw[fill=gray,opacity=0.6] (1,0.5,0) -- (0,0.5,1) -- (-1,0.5,0) -- (0,0.5,-1) -- cycle;	
    \draw (1.7,0.7,.4+0.5) node{$\mu_{2}$};		
    \draw (-.2+0.1,0.6,1.5) node{$\mu_{1}$};	
    \draw (-.1,1.7,0) node{$\mu_{12}$};	
    \draw (-1,0.5,0) node{$\bullet$};		
    \draw (0,0.5,1) node{$\bullet$};		
    \draw (0,0.5,-1) node{$\bullet$}; 	
    \draw (1,0.5,0) node{$\bullet$};		
    \draw (0,0.1,0.5) node{\small $(1,0,0)$};		
    \draw (-1.2,0.5,0.5) node{\small $(0,-1,0)$};		
    \draw (0.5,0.55,-1) node{\small $(-1,0,0)$};		
    \draw (0.2,-0.4,-2.8) node{\small $(0,1,0)$};	
\end{tikzpicture}  
\\
(a)	& (b)
\end{tabular}
\caption{Marginal polytope of the BC model with two nodes and one edge, which in $\mathbb{R}^4$ is a connected figure with the vertices connected to each other. In (a) is the space for the level 2, that is, $\phi_4(x)=x_1^2+x_2^2=2$, and in (b) $\phi_4(x)=1$; there is another level where $\phi_4(x)=0$ and is the point $(0,0,0,0)$. }
\label{fig:marginal-polytope}
\end{figure}

The means $\mu_s$ and $\mu_{st}$ in the interior of $\mathbb{M}$ (i.e., not on the nodes, edges or faces of the polytope) are associated with a unique $\theta_s$ and $\theta_{st}$ in the parameter space $\Omega$ and $\nabla_s\log Z_\theta(\mu_s)=\theta_s$ \citep{McCulloch:1988}. Since the model is full and minimal, the natural parameter space is $\mathbb{R}^4$, and the mean parameter space $\mathbb{M}$ is a bounded convex polytope in $\mathbb{R}^4$ \citep[see, e.g.,][Corollary 6.16, Chapter 1]{Lehmann:1983}.


\section{Pseudo-likelihood}\label{sec:pseudo-likelihood}\noindent
It is immediate that the probability $p_\theta$ of the BC model is computationally impeded by the normalisation constant in equation (\ref{eq:normalising-constant}) because it sums over all configurations of $x\in \{-1,0,+1\}^V$. For example, with 20 nodes in a graph we require a sum over $3^{20}$ (nearly 3.5 billion) elements. We therefore require simplifications to be able to approximate the normalisation constant. 

An approximation to the joint distribution is used such that the normalisation constant becomes a product. One possibility to obtain such a product is to assume that the graph is empty, and we then obtain the product distribution $p_\theta(x)=p_\theta(x_1)p_\theta(x_2)\cdots p_\theta(x_m)$ \citep[][Chapter 5]{Wainwright:2008}. This solves the computational problem because now we can evaluate each variable separately. However, this so-called naive approach corresponds to assuming independence between variables and may significantly misrepresent the true joint distribution (i.e., it is a lower bound much smaller than the joint distribution). Another approach is to approximate the joint distribution with the product of the full conditional distributions \citep{Besag:1974}. This is often referred to as the pseudo-likelihood and has been shown to be consistent \citep[see, e.g.,][and see Appendix \ref{app:consistency-pseudo-likelihood} for more details]{Nguyen:2017}. This pseudo-likelihood approach is the one we take here. 

For any variable $X_s$, we obtain its conditional distribution given all remaining variables $X_{t\ne s}:= X_{V\backslash \{s\}}$. This conditional distribution is denoted by $p_\theta(x_s\mid x_{t\ne s})$, where we collect the parameters $\tau$, $\sigma_{st}$ for $t\ne s$ and $\alpha^2$ in the vector $\theta$. (Note that we ignore the subscript $s$ to indicate the dependent variable; we do this for convenience of notation.) Then the joint distribution is approximated by the product of all the conditional distributions.
\begin{align}\label{eq:pseudo-likelihood}
p_{\theta}(x)\approx p_\theta(x_1\mid x_{t\ne 1})p_\theta(x_2\mid x_{t\ne 2})\cdots p_\theta(x_m\mid x_{t\ne m})
\end{align}
The product of conditional distributions is called the pseudo-likelihood. 

For the BC model, the conditional distribution of $X_s$ given the remaining variables $X_{t\ne s}$ is
\begin{align}\label{eq:conditional-distribution}
p_\theta(x_s\mid x_{t\ne s}) =  \frac{\exp ( \tau_s x_s + x_s\sum_{t\ne s} \sigma_{st} x_t -\alpha^2 x_s^2) }
					{ 1 + 2\cosh(\tau_s + \sum_{t\ne s} \sigma_{st} x_t)\exp(-\alpha^2) }
\end{align}
Note that $\beta$ is absorbed into the other parameters, and does not appear explicitly. Only the relative magnitudes of these parameters matter. This conditional distribution is also used for the Gibbs sampler, that we use to simulate observations from the model, in which the full conditional in equation (\ref{eq:conditional-distribution}) for each variable is used to form a Markov chain to obtain the joint distribution \citep{Haggstrom:2002}. For the Gibbs sampler we include the inverse temperature; we describe the details of the Gibbs sampler in Appendix \ref{app:gibbs-sampler}.  

Because the product of conditionals in equation (\ref{eq:pseudo-likelihood}) results in a similar distribution as the joint distribution in equation (\ref{eq:joint-distribution}), we can apply the theory of exponential families to the pseudo-likelihood setting (see Appendix \ref{app:pseudo-likelihood} for more details). For instance, we easily find that the conditional expectation of $X_s$ given $x_{t\ne s}$ is obtained by 
\begin{align*}
\nabla_s \log Z_\theta = \mathbb{E}(X_s\mid x_{t\ne s}) 
=
\frac{2\sinh(\tau_s + \sum_{t\ne s} \sigma_{st} x_t)\exp(-\alpha^2)} { 1 + 2\cosh(\tau_s + \sum_{t\ne s} \sigma_{st} x_t)\exp(-\alpha^2) } 
\end{align*}
We obtain a similar result when using mean field theory (see Appendix \ref{app:mean-field}). 
It is also possible to derive second-order derivatives, which gives the conditional covariance matrix of the sufficient statistics \citep[][and see Appendix \ref{app:pseudo-likelihood}]{McCulloch:1988,Wainwright:2008,Nguyen:2017}.

\subsection{Point estimation}\label{sec:estimation}\noindent
Let $X=(X_1,X_2,\ldots, X_m)$ be the $n\times m$ matrix with $X_s$ the vector associated to node $s\in \{1,2,\ldots,m\}$ for the observations $(x_{s,1},x_{s,2},\ldots, x_{s,n})$ of dimension $n$. We use the log pseudo-likelihood function %
\begin{align}\label{eq:pseudo-likelihood-estimation}
\ell_\theta^n(X) = -\frac{1}{n}\sum_{i=1}^n\log p_\theta(x_{s,i}\mid x_{t\ne s,i})
\end{align}
for estimation \citep{Besag:1974}. Minimisation of $\ell_\theta^n$ leads to a unique estimate if the parameter $\theta$ is in the interior of $\Omega$ (equivalently, if $\mu$ is in the interior of $\mathbb{M}$). We call an estimate obtained with the log pseudo-likelihood function in equation (\ref{eq:pseudo-likelihood-estimation}) a pseudo-likelihood estimate and denote it by $\hat{\theta}$.
It has been shown that under relatively mild conditions the pseudo-likelihood estimate $\hat{\theta}$ converges to the true parameter $\theta_0$ when $n\to \infty$ \citep[][Theorem 1, see also Appendix \ref{app:consistency-pseudo-likelihood}]{Nguyen:2017}. 

In many cases, there are relatively many edges compared to the number of observations. In such high-dimensional settings, we may use a penalty (Lagrangian) for the pseudo-likelihood \citep{Buhlmann:2014b}. Often the lasso penalty is used \citep[least absolute shrinkage and selection operator (lasso),][]{Hastie:2015}, introduced by \citet{Tibshirani:1996}, and makes the solution unique with respect to the assumption of sparsity. Sparsity refers to the assumption that relatively few (say no more than $5\%$) of all possible non-zero coefficients, are non-zero \citep{Hastie:2015}. Several theoretical results of the lasso with implications for small samples have been obtained in \citet{Wainwright:2009}, and has been applied to logistic regression \citep{Ravikumar:2010,Borkulo:2014,Waldorp:2019} and inference  \citep[$p$-values and confidence intervals,][]{Geer:2014,Waldorp:2024}. 

The lasso estimate $\hat{\theta}$ is obtained by minimising
\begin{align}\label{eq:lasso}
\ell_\theta^n(X) + \lambda ||\sigma ||_1 = -\frac{1}{n}\sum_{i=1}^n\log p_\theta(x_{s,i} \mid x_{t\ne s,i}) + \lambda \sum_{t\ne s} |\sigma_{st}|
\end{align}
where the sum of absolute values of the parameters $\sigma_{st}$ is included in the Lagrangian with Lagrangian parameter $\lambda$. The parameter $\lambda$ can be obtained by cross-validation \citep{Hastie:2015} or a type of Bayesian information criterion \citep[extended BIC,][]{Foygel:2010,Haslbeck:2020}.

The lasso is convenient because it also has the property of selecting the parameters \citep{Buhlmann:2011}. The statistical guarantees of low false positive rate and reasonable true positive rate come at the price of several assumptions, of which the sparsity assumption is the most well-known:  for each node the number of non-zero parameters (connections) is sparse, i.e., $s_0$ is in the order of $\sqrt{\smash[b]{n/\log m}}$, where $m$ is the number of nodes $|V|=m$. These assumptions lead to the guarantee for consistency of $\hat{\theta}$ to $\theta_0$ (the true value) if for each node, $n$ is at least $2\gamma m\log (1-\gamma)m)$, with $\gamma\in (0,1)$ and $\lambda$ is of the order $\sqrt{\smash[b]{\log (m)/n}}$ \citep[][Section 6.7]{Wainwright:2009,Buhlmann:2011}. For instance, in a graph with 30 nodes  we require at least $n=52$ observations for accurate false positive rates. More details about the assumptions and statistical guarantees are in Appendix \ref{app:pseudo-likelihood}. 

\subsection{Confidence intervals}\label{sec:standard-errors}\noindent
To quantify the uncertainty of the parameter estimates $\hat{\theta}$ we use confidence intervals. We use the normal approximation to confidence intervals \citep{Cox:2007}. For parameter estimate $\hat{\theta}_{s}$ (with $s$ referring to any of the BC model parameters) with standard error $\sqrt{\smash[b]{\Gamma_{ss}/n}}$,  where $\Gamma$ is the asymptotic variance matrix of the estimator,  we obtain a $95\%$ confidence interval as
\begin{align*}
\left(\hat{\theta}_{s} - 1.96\sqrt{\Gamma_{ss}/n}, \phantom{i} \hat{\theta}_{s} + 1.96\sqrt{\Gamma_{ss}/n}\right)
\end{align*}
In obtaining accurate confidence intervals, i.e., confidence intervals with coverage of $95\%$, we face two problems. First, in using the pseudo-likelihood to obtain the point estimates (to avoid the computational issue of calculating the normalisation constant) we have misspecified the model. Hence, the usual way of computing standard errors \citep[using the inverse of the second-order derivatives, the Hessian,][]{Nelder:1972,McCulloch:1988} is not valid. Second, in using the lasso, the sampling distribution of the lasso estimate is not normal; it is not continuous and has probability mass of one at zero if the true parameter is zero \citep{Buhlmann:2014b}. We discuss each problem briefly. We give more details of the solution to these two issues in Appendix \ref{app:standard-errors}

To determine the standard errors from the pseudo-likelihood we need to take into account that the model specification is incorrect \citep{Schmid:2023}. The common result that the product of first-order derivatives (Fisher information) cancels with the inverse of the second-order derivative (Hessian), does not hold when the model is misspecified \citep[e.g.,][Example 5.25]{White81,Vandervaart98}. Hence, when there is misspecification (as is the case for the pseudo-likelihood) we cannot use the standard Fisher information or the Hessian for the standard errors. 
This implies that the correct standard errors must be obtained from the so-called sandwich estimate, where both the Hessian and Fisher information are used together \citep{White81,Schmid:2023}.

The sampling distribution of the lasso estimate is not a Gaussian distribution, nor is it continuous \citep{Potscher:2009,Buhlmann:2014b,Waldorp:2024}. Therefore, obtaining standard errors and constructing confidence intervals is not straightforward \citep{Dezeure:2015,Williams:2020b}. One possibility that leads to accurate confidence intervals is to use the desparsified lasso \citep[][and see Appendix \ref{app:standard-errors}]{Geer:2014,Buhlmann:2013,Waldorp:2024}. The distribution of the desparsified lasso also leads to the sandwich estimate with the Hessian and first-order derivatives \citep{Geer:2014,Buhlmann:2013}, similar to the sandwich of the pseudo-likelihood. For the case where we have the number of observations that is in the same order of the number of parameters we use the shrinkage estimate of the second-order derivatives to obtain accurate confidence intervals. In Appendix \ref{app:dLshrink} we prove that using a shrinkage parameter of $1/n^{1+\gamma}$ for $\gamma>0$ leads to good coverage. 


\section{Numerical illustrations}\label{sec:numerical-illustrations}\noindent

\subsection{Synthetic data}\noindent
To illustrate the estimation of the BC parameters $\theta$, consisting of connection parameters $\sigma_{st}$, thresholds $\tau_s$ and the zero-controlling parameter $\alpha^2$, we create random networks, generate data with the fixed parameters using the Gibbs sampler, and then estimate the parameters with the pseudo-likelihood method. We are mostly interested in the true and false positive rates of the lasso (the ability to recover the non-zero edge weights from data, treating the true sparsity pattern as ground truth) and the confidence intervals (coverage rates). 

Networks of 10, 20 and 30 nodes are created where the edges are placed randomly between two nodes with probability of $p_e=0.3$ (Erd\"{o}s-R\'{e}nyi networks $\mathcal{G}(m,p_e)$). So for 10 node networks we have on average 13.5 edges. Then using the BC model we obtain data using the Gibbs sampler with parameters
\begin{align*}
\tau_s = 1.2, \qquad \sigma_{st}=1, \qquad \alpha^2=1.2
\end{align*}
 for all $s\in V$ and $(s,t)\in E$. For the data sets we vary the sample size from 50 to 260. We generated 100 datasets for each number of observations, and then estimated the parameters for each of these 100 datasets to determine the accuracy of the estimates. 
 
Network parameters are estimated with the pseudo-likelihood method described above. 
To determine the accuracy of the pseudo-likelihood and lasso combination, we consider several measures. We consider both the accuracy of identification of non-zero parameters with the lasso and the coverage of the confidence intervals. We average results over parameters and so use the symbol $\hat{\theta}$ for any of the parameters in the BC model, and $\theta_0$ for the true value. 

To determine the accuracy of identification of non-zero parameters we use the true positive rate (TPR)
\begin{align*}
\text{true positive rate} = \mathbb{P}_n(\hat{\theta}\ne 0\mid \theta_0\ne 0 )= \frac{ \#\{ \text{identified and true} \} } { \#\{ \text{true}\} } 
\end{align*}
where $\mathbb{P}_n$ refers to the probability (proportion over simulations with sample size $n$) determined from the simulations. 
And the false positive rate
\begin{align*}
\text{false positive rate} = \mathbb{P}_n(\hat{\theta}\ne 0\mid \theta_0 = 0) = \frac{ \#\{ \text{identified and false} \} } { \#\{ \text{false}\} }
\end{align*}
For the accuracy of the $95\%$ confidence intervals we consider the coverage
\begin{align*}
\text{coverage} = \mathbb{P}_n( \text{true parameter in $95\%$ confidence interval}) 
\end{align*}
Finally, for the accuracy of the point estimates we consider the bias
\begin{align*}
\text{bias} = | \mathbb{E}_n( \hat{\theta}) - \theta_0  |
\end{align*}
where $\mathbb{E}_n$ refers to the average determined from the simulations. 

\subsubsection{Results}\label{sec:results}\noindent
In Figure \ref{fig:tpr-fpr-bias}(a) the true and false positive rates (TPR and FPR) of the edge parameters for networks with either 10, 20 or 30 nodes, as a function of sample size. Corresponding to the theory of lasso estimation in Section \ref{sec:estimation}, the FPR should not be higher than the 0.05 level for any number of observations, even for 30 nodes. It can be seen in Figure \ref{fig:tpr-fpr-bias}(a) that the FPR is not higher than 0.05 from $n=50$ observations or for any number of nodes in the network. So, the probability of a false positive remains low. The TPR increases  with $n$, and depends on the number of nodes in the network. 
In Figure \ref{fig:tpr-fpr-bias}(b) we see that the coverage of the $95\%$ confidence intervals is close to the nominal $0.95$ for the sandwich estimator (using the shrinkage estimate with parameter $\rho=1/n^{5/4}$, see Appendix \ref{app:dLshrink}) at any number of observations, but too low for the $95\%$ confidence intervals based on the Fisher information only. Figure \ref{fig:coverage}(b) shows the standard error estimates based on either the sandwich or the Fisher information. It clearly shows that the sandwich standard errors are reasonable, but the Fisher estimates are too low. 
\begin{figure}[t]\centering
\begin{tabular}{@{\hspace{-1em}} c @{\hspace{1em}} c @{\hspace{-1em}} c}
		&\pgfimage[width=.52\textwidth]{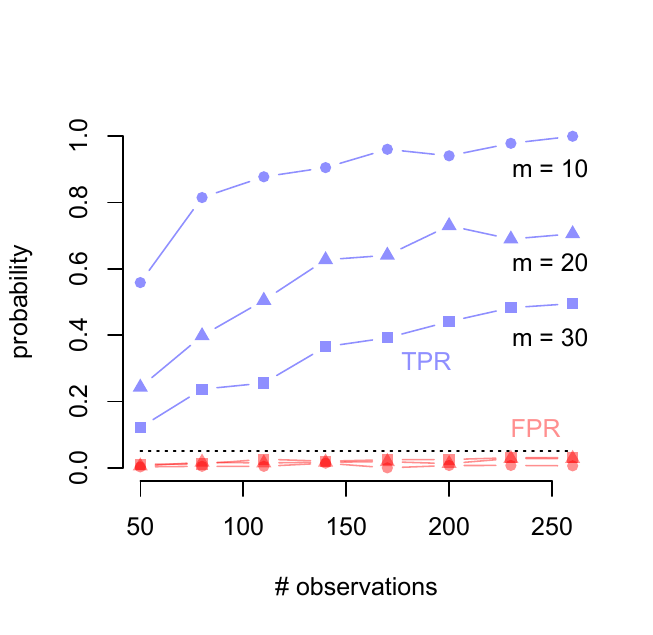}		&\pgfimage[width=.52\textwidth]{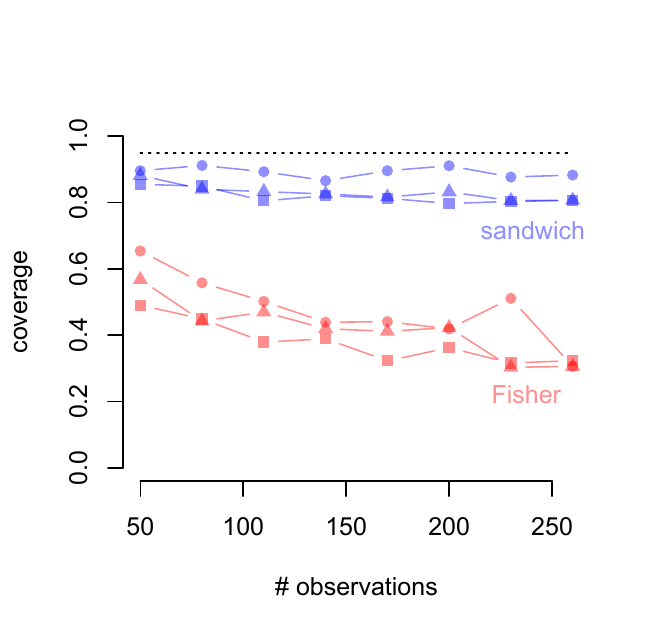}\\
	&	(a)		&(b)
\end{tabular}
\caption{Results for random networks of 10, 20 and 30 nodes with edge probability $0.3$, $\sigma_{st}=1$, $\alpha^2=1.2$ and $\tau=1.2$. Three different network sizes are shown: $\circ$ $m=10$, $\triangle$ $m=20$ and $\Box$ $m=30$. In (a) are the true (blue) and false (red) positive rates averaged over all edge parameters. In (b) is the coverage of confidence intervals using either the sandwich (blue) or the Fisher information (red).}
\label{fig:tpr-fpr-bias}
\end{figure}

In Figure \ref{fig:coverage}(a) it can be seen that the bias is low and  decreases slowly, for the connectivity parameters $\sigma$, for the $\alpha^2$ parameters controlling the zeros, and for the threshold parameters $\tau$.
\begin{figure}\centering
\begin{tabular}{c @{\hspace{0em}} c }
		\pgfimage[width=.52\textwidth]{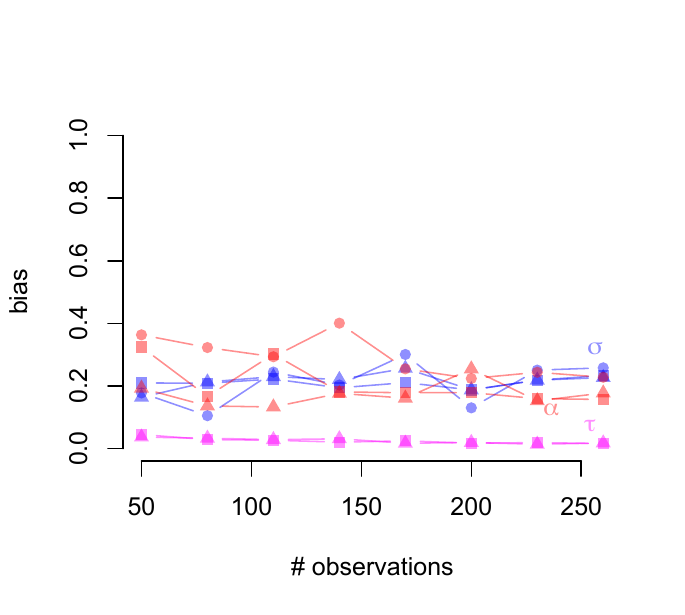}		&\pgfimage[width=.52\textwidth]{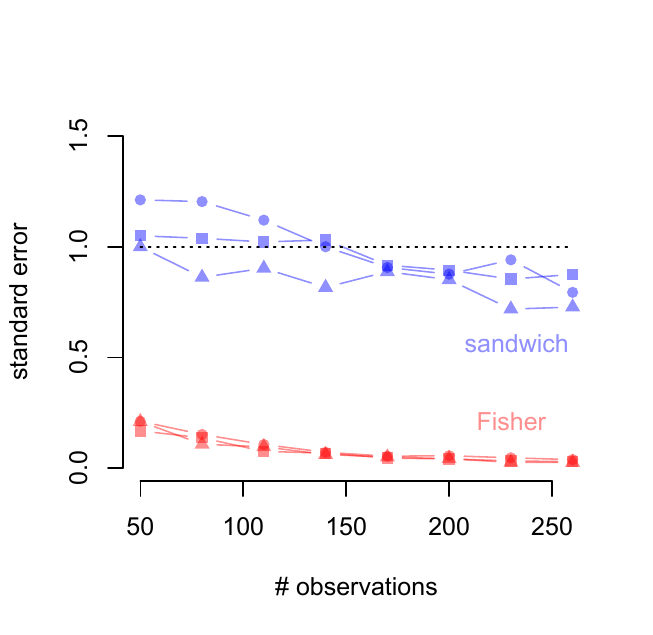}\\
	(a)		&(b)
\end{tabular}
\caption{Results for networks of size 10 (circles), 20 (triangles) and 30 (squares), with edge probability $0.3$, $\sigma_{st}=1$, $\alpha^2=1.2$ and $\tau=1.2$. In (a) is the bias for each of the networks separately and each of the parameters separately. And in (b) are the standard errors of the estimates (averaged over all parameters) using either the sandwich and Fisher information.}
\label{fig:coverage}
\end{figure}
%

%

\subsection{Application to attitudes on political issues}\label{sec:application-to-data}\noindent
To illustrate some interesting properties of the model we apply the Blume-Capel estimation to attitudes toward societal issues related voting behaviour. The data were collected on the online platform \textit{Stemwijzer}, which assists individuals in making a voting decision in the Dutch national election. Individuals indicate their agreement or disagreement with several political statements and then they receive an indication how well-aligned their beliefs are with the different political parties. The data were collected shortly before the Dutch national election in 2024 by the company Prodemos, which granted us access to the data. We obtained 10,000 observations (though we do not know if all were from different individuals) for 30 variables. We selected the 19 most interesting variables (see Table \ref{tab:stemwijzer-variables}). 

We applied the Blume-Capel model to obtain an estimate of the network parameters $\sigma_{st}$, the threshold (external field) parameters $\tau_s$, and we obtain estimates of the $\alpha^2_s$ parameters. Note that the $\alpha^2_s$ are defined for each node \textit{separately}. This is different from the original model but is a natural consequence of the pseudo-likelihood algorithm used to obtain estimates for each node separately.
To obtain the estimates, we set the penalty parameter $\lambda$ for the lasso to $\sqrt{\smash[b]{\log (m)/n}}$, which amounts to $0.0197$, with $m=19$ nodes and $n=10,000$ observations. The shrinkage parameter to obtain standard errors and confidence intervals was set to $1/n^{5/4}$, as in the simulations. We ran the algorithm also using a subsampling procedure to test for reliability, with results very similar to the result we present here (see Appendix \ref{app:application-additional}). 

Figure \ref{fig:stem-graph-sigma}(a) shows the graph obtained from the algorithm. The thickness of the edges reflects their coefficients, which are given as desparsified values in Figure \ref{fig:stem-graph-sigma}(b) along with their confidence intervals (see Appendix \ref{app:dLshrink}). The confidence intervals are small, as expected with a large sample size. The network in Figure \ref{fig:stem-graph-sigma}(a) reflects the selection by the lasso algorithm and the threshold based on sparsity $\sqrt{\smash[b]{\log(m)/n}}$. The network shows several noteworthy results. First, we only observe positive edges, which aligns well with the notion that attitude networks strive for consistency \citep{Dalege:2016}. Second, we observe some clustering with nodes that tap into similar topics that are closely connected. For example, nodes 13, 14, 17, and 18 all represent statements related to immigration, and these nodes form a tight cluster. Third, except for node 15, all nodes have at least one edge connecting them to other nodes, indicating that information can flow through the whole network. 
\begin{figure}\centering
\begin{tabular}{c @{\hspace{0em}} c }
		\raisebox{2em}{\pgfimage[width=.52\textwidth]{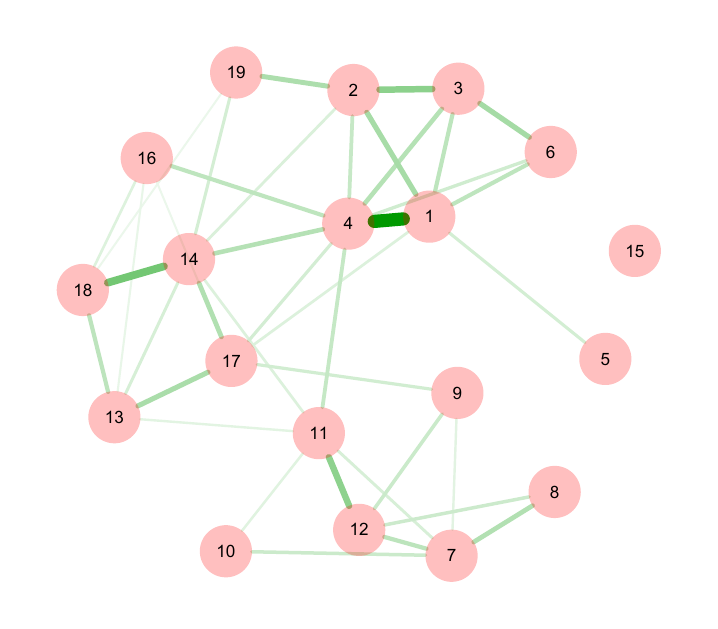}}		&\pgfimage[width=.52\textwidth]{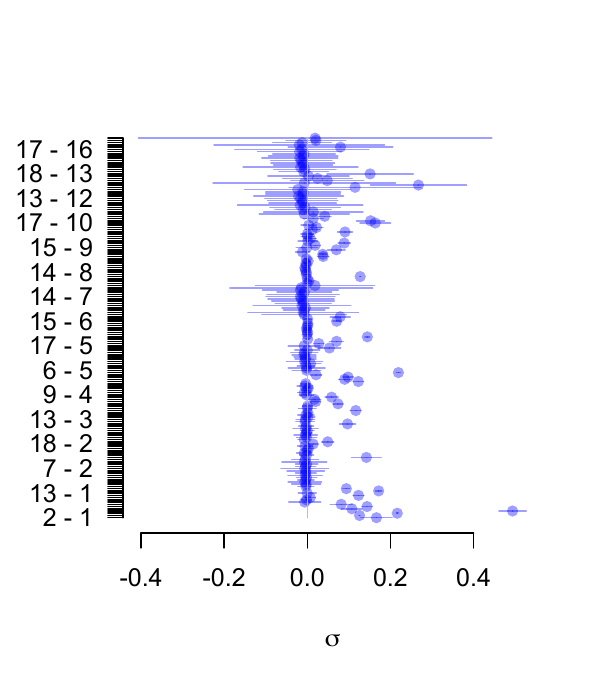}\\
	(a)		&(b)
\end{tabular}
\caption{In (a) is the network estimate using the BC model with the lasso algorithm, thresholded using $\sqrt{\smash[b]{\log(m)/n}}$; edges below this threshold are set to 0. In (b) are the confidence intervals of the edge parameters $\sigma$ based on the shrinkage estimate described in Appendix \ref{app:dLshrink}. In Table \ref{tab:stemwijzer-variables} are the labels for the nodes (variables) in the data.}
\label{fig:stem-graph-sigma}
\end{figure}

Figure \ref{fig:stem-alpha}(a) shows the confidence intervals for the $\alpha^2_s$ parameters. And in Figure \ref{fig:stem-alpha}(b) we see the regression of the frequency of zeros on the value of $\alpha^2_s$ for each of the 19 nodes. The high correlation, here $r=0.98$, between $\alpha^2_s$ and the frequency of zeros, suggests that the estimate of $\alpha^2_s$ is an accurate representation of the frequency of zeros. Note that the values of $\alpha_s^2$ are negative here. The interpretation as in the original BC model remains: the lower the value of $\alpha_s^2$ the lower the frequency of zeros. We obtain a similar result when setting the penalty parameter to $\lambda=0$, obtaining unregularised estimates. 
\begin{figure}\centering
\begin{tabular}{c @{\hspace{0em}} c }
\pgfimage[width=.52\textwidth]{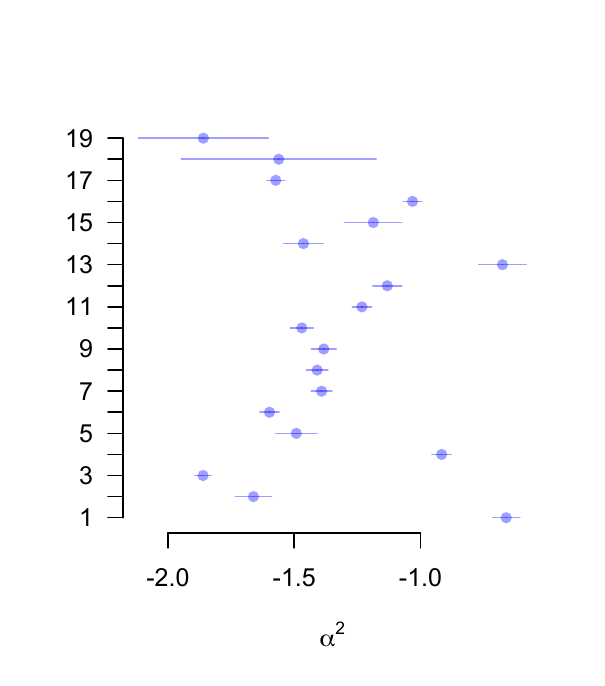}		&\pgfimage[width=.52\textwidth]{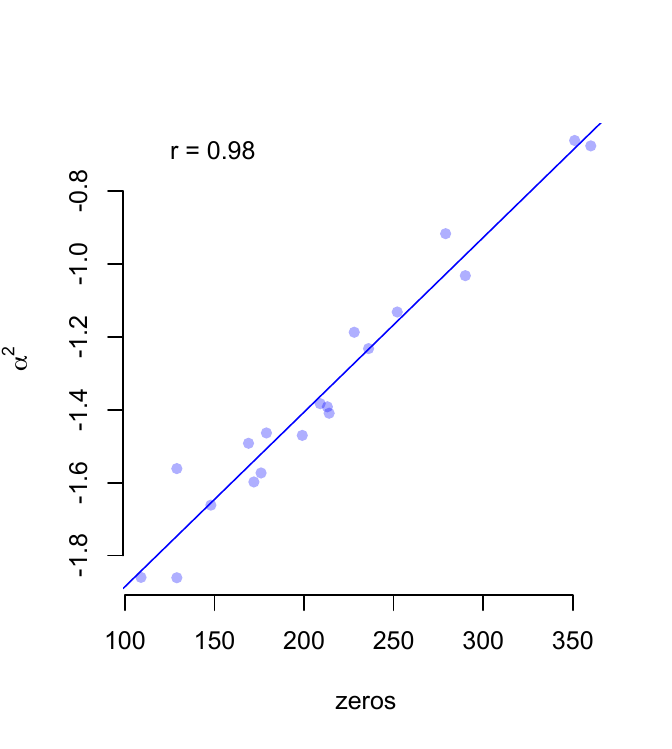}\\
	(a)		&(b)
\end{tabular}
\caption{In (a) are the desparsified estimates and the corresponding confidence intervals. In (b) are the number of zeros in each of the variables and the estimate of the parameter $\alpha^2_s$.   }
\label{fig:stem-alpha}
\end{figure}
\begin{table}
\begin{tabular}{l @{\hspace{3em}} l}
\midrule
1       &Vee (Cattle)\\[-0.5em]
2       &Lagere prijs benzine (Low gas prices)\\[-0.5em]
3       &Vliegbelasting (Airplane travel tax)\\[-0.5em] 
4       &Stikstof (Carbon emissions more expensive)\\[-0.5em] 
5       &Bouwen op landbouwgrond (Building on farm land)\\[-.5em] 
6       &Vuurwerk (Fireworks allowed)\\[-0.5em] 
7       &Eigen risico zorgverzekering (Insurance more expensive)\\[-0.5em] 
8       &OV voor 65-plussers (Free public transport for 65+)\\[-0.5em] 
9       &Kinderopvang zonder winstoogmerk (Daycare without profit)\\[-0.5em] 
10      &Afschaffen kostendelersnorm (Reduce benefits)\\[-0.5em] 
11      &Hoger minimumloon (Higher minimum wages)\\[-0.5em] 
12      &Huurstijging (Increase rent of housing)\\[-0.5em] 
13      &Visum Surinamers (Visa for people from Suriname)\\[-0.5em] 
14      &Gezinshereniging (Reunite immigrant families)\\[-0.5em] 
15      &Controle op religieuze groepen (Control on religious groups)\\[-0.5em] 
16      &Uitbreiding Europese Unie (Expansion European Union)\\[-0.5em] 
17      &Risicoprofilering op basis van nationaliteit (Risk assessment based on nationality)\\[-0.5em] 
18      &Ontwikkelingshulp en asielzoekers (More money for International development) \\[-0.5em] 
19      &Minimumstraffen (Minimum sentences)\\
\midrule
\end{tabular}
\caption{Numbers corresponding to the questions in the "stemwijzer" questionnaire in Dutch, with English translation. The numbers correspond to those in Figure \ref{fig:stem-graph-sigma}.}
\label{tab:stemwijzer-variables}
\end{table}
\section{Discussion}\label{sec:discussion}\noindent
We introduced the Blume-Capel model here as an interesting model that could be applied to data for inference on its parameters. The BC model is more interesting than the Ising model because (i) there are three stable states instead of two, (ii) there is a neutral position for the answer categories, and (iii) this neutral position is accompanied by a specific parameter. 

To obtain estimates we obtained the pseudo-likelihood equations. We used the lasso to deal with high dimensional data (many parameters, few observations), and used the sandwich estimator to obtain accurate confidence intervals for the lasso estimates. We showed that the combination of a shrinkage estimate for the second-order derivatives and the desparsified lasso yield correct confidence intervals. These results were confirmed by numerical simulations. 

The application to voting behaviour data made clear why the BC model is informative beyond the Ising model. Because we had $\{-1,0,1\}$ data, the BC model fits well with these response categories, and the neutral parameter of the position was shown to be highly correlated to the frequency of 0s in the data. The parameter $\alpha^2$, therefore, picks up the pattern in the data responsible for the frequency of 0s. We, therefore, propose that the $\alpha^2$ parameter be called neutrality or caution parameter, depending on the context.


\begin{thebibliography}{}

\bibitem[Ak{\i}nc{\i}, 2016]{Akinci:2016}
Ak{\i}nc{\i}, {\"U}. (2016).
\newblock Crystal field dilution in s-1 blume capel model: Hysteresis
  behaviors.
\newblock {\em Physics Letters A}, 380(14-15):1352--1357.

\bibitem[Albert and Swendsen, 2014]{Albert:2014}
Albert, J. and Swendsen, R.~H. (2014).
\newblock The inverse ising problem.
\newblock {\em Physics Procedia}, 57:99--103.

\bibitem[Barndorff-Nielsen and Cox, 1994]{Barndorff-Nielsen:1994}
Barndorff-Nielsen, O. and Cox, D. (1994).
\newblock {\em Inference and asymptotics}.
\newblock Chapman and Hall.

\bibitem[Barrat et~al., 2008]{Barrat:2008}
Barrat, A., Barthelemy, M., and Vespignani, A. (2008).
\newblock {\em Dynamical processes on complex networks}.
\newblock Cambridge University Press.

\bibitem[Besag, 1974]{Besag:1974}
Besag, J. (1974).
\newblock Spatial interaction and the statistical analysis of lattice systems.
\newblock {\em Journal of the Royal Statistical Society. Series B
  (Methodological)}, 36(2):192--236.

\bibitem[Blume, 1966]{Blume:1966}
Blume, M. (1966).
\newblock Theory of the first-order magnetic phase change in uo.
\newblock {\em Physical Review}, 141(2):517.

\bibitem[Blume et~al., 1971]{Blume:1971}
Blume, M., Emery, V.~J., and Griffiths, R.~B. (1971).
\newblock Ising model for the $\lambda$ transition and phase separation in he
  3-he 4 mixtures.
\newblock {\em Physical review A}, 4(3):1071.

\bibitem[Bouabci and Carneiro, 2000]{Bouabci:2000}
Bouabci, M.~B. and Carneiro, C. (2000).
\newblock Random-cluster representation for the blume--capel model.
\newblock {\em Journal of Statistical Physics}, 100:805--827.

\bibitem[B{\"u}hlmann, 2013]{Buhlmann:2013}
B{\"u}hlmann, P. (2013).
\newblock Statistical significance in high-dimensional linear models.
\newblock {\em Bernoulli}, 19(4):1212--1242.

\bibitem[B{\"u}hlmann et~al., 2014]{Buhlmann:2014b}
B{\"u}hlmann, P., Kalisch, M., and Meier, L. (2014).
\newblock High-dimensional statistics with a view toward applications in
  biology.
\newblock {\em Annual Review of Statistics and Its Application}, 1:255--278.

\bibitem[B{\"u}hlmann and van~de Geer, 2011]{Buhlmann:2011}
B{\"u}hlmann, P. and van~de Geer, S. (2011).
\newblock {\em Statistics for High-Dimensional Data: Methods, Theory and
  Applications}.
\newblock Springer.

\bibitem[Capel, 1966]{Capel:1966}
Capel, H.~W. (1966).
\newblock On the possibility of first-order phase transitions in ising systems
  of triplet ions with zero-field splitting.
\newblock {\em Physica}, 32(5):966--988.

\bibitem[Claeskens and Hjort, 2008]{Claeskens:2008}
Claeskens, G. and Hjort, N. (2008).
\newblock {\em Model selection and model averaging}.
\newblock Cambridge University Press, Cambridge.

\bibitem[Cox, 2007]{Cox:2007}
Cox, D. (2007).
\newblock {\em Principles of statistical inference}.
\newblock Cambridge University Press.

\bibitem[Dalege et~al., 2016]{Dalege:2016}
Dalege, J., Borsboom, D., Van~Harreveld, F., Van~den Berg, H., Conner, M., and
  Van~der Maas, H.~L. (2016).
\newblock Toward a formalized account of attitudes: The causal attitude network
  (can) model.
\newblock {\em Psychological review}, 123(1):2.

\bibitem[Dezeure et~al., 2015]{Dezeure:2015}
Dezeure, R., B{\"u}hlmann, P., Meier, L., and Meinshausen, N. (2015).
\newblock High-dimensional inference: confidence intervals, p-values and
  r-software hdi.
\newblock {\em Statistical science}, pages 533--558.

\bibitem[Epskamp, 2015]{Epskamp:2015aa}
Epskamp, S., M. G. W. L. J. . B.~D. (2015).
\newblock Network psychometrics.
\newblock In Irwing, P., Hughes, D., and Booth, T., editors, {\em Handbook of
  Psychometrics}. John Wiley \& Sons, New York.

\bibitem[Fernandez et~al., 2016]{Fernandez:2016}
Fernandez, M.~A., Korutcheva, E., and de~la Rubia, F.~J. (2016).
\newblock A 3-states magnetic model of binary decisions in sociophysics.
\newblock {\em Physica A: Statistical Mechanics and its Applications},
  462:603--618.

\bibitem[Ferri et~al., 2022]{Ferri:2022}
Ferri, I., P{\'e}rez-Vicente, C., Palassini, M., and D{\'\i}az-Guilera, A.
  (2022).
\newblock Three-state opinion model on complex topologies.
\newblock {\em Entropy}, 24(11):1627.

\bibitem[Foygel and Drton, 2010]{Foygel:2010}
Foygel, R. and Drton, M. (2010).
\newblock Extended bayesian information criteria for gaussian graphical models.
\newblock {\em Advances in neural information processing systems}, 23.

\bibitem[Graham and Grimmett, 2006]{Graham:2006}
Graham, B. and Grimmett, G. (2006).
\newblock Random-cluster representation of the blume--capel model.
\newblock {\em Journal of statistical physics}, 125:283--316.

\bibitem[Grimmett, 2022]{Grimmett:2022}
Grimmett, G. (2022).
\newblock The potts and random-cluster models.
\newblock In {\em Handbook of the Tutte Polynomial and Related Topics}, pages
  378--394. Chapman and Hall/CRC.

\bibitem[Grimmett, 2006]{Grimmett:2006b}
Grimmett, G.~R. (2006).
\newblock {\em The random-cluster model}, volume 333.
\newblock Springer Science \& Business Media.

\bibitem[H{\"a}ggstr{\"o}m, 2002]{Haggstrom:2002}
H{\"a}ggstr{\"o}m, O. (2002).
\newblock {\em Finite Markov chains and algorithmic applications}, volume~52.
\newblock Cambridge University Press.

\bibitem[Haslbeck and Waldorp, 2020]{Haslbeck:2020}
Haslbeck, J.~M. and Waldorp, L.~J. (2020).
\newblock mgm: Structure estimation for time-varying mixed graphical models in
  high-dimensional data.
\newblock {\em Journal of Statistical Software}, 93(8):1--46.

\bibitem[Hastie et~al., 2015]{Hastie:2015}
Hastie, T., Tibshirani, R., and Wainwright, M. (2015).
\newblock {\em Statistical learning with sparsity: the lasso and
  generalizations}.
\newblock CRC Press.

\bibitem[Hirsch et~al., 2004]{Hirsch:2004}
Hirsch, M., Smale, S., and Devaney, R. (2004).
\newblock {\em Differential equations, dynamical systems, and an introduction
  to chaos}.
\newblock Academic Press, 2nd edition.

\bibitem[Holmgren, 1996]{Holmgren:1996}
Holmgren, R. (1996).
\newblock {\em A first course in discrete dynamical systems}.
\newblock Springer Science \& Business Media.

\bibitem[Javanmard and Montanari, 2014]{Javanmard:2014}
Javanmard, A. and Montanari, A. (2014).
\newblock Confidence intervals and hypothesis testing for high-dimensional
  regression.
\newblock Technical report, arXiv:1306.317.

\bibitem[Kaufman and Kanner, 1990]{Kaufman:1990}
Kaufman, M. and Kanner, M. (1990).
\newblock Random-field blume-capel model: Mean-field theory.
\newblock {\em Physical Review B}, 42(4):2378.

\bibitem[Kindermann et~al., 1980]{Kindermann:1980}
Kindermann, R., Snell, J.~L., et~al. (1980).
\newblock {\em Markov random fields and their applications}, volume~1.
\newblock American Mathematical Society Providence, RI.

\bibitem[Lehmann, 1983]{Lehmann:1983}
Lehmann, E.~L. (1983).
\newblock {\em Theory of point estimation}.
\newblock Wiley and Sons, New York.

\bibitem[Maathuis et~al., 2018]{Maathuis:2018}
Maathuis, M., Drton, M., Lauritzen, S., and Wainwright, M. (2018).
\newblock {\em Handbook of graphical models}.
\newblock CRC Press.

\bibitem[Marsman et~al., 2018]{Marsman:2018}
Marsman, M., Borsboom, D., Kruis, J., Epskamp, S., van Bork, R., Waldorp, L.,
  Maas, H. v.~d., and Maris, G. (2018).
\newblock An introduction to network psychometrics: Relating ising network
  models to item response theory models.
\newblock {\em Multivariate Behavioral Research}, 53(1):15--35.

\bibitem[McCulloch, 1988]{McCulloch:1988}
McCulloch, R.~E. (1988).
\newblock Information and the likelihood function in exponential families.
\newblock {\em The American Statistician}, 42(1):73--75.

\bibitem[Nelder and Wedderburn, 1972]{Nelder:1972}
Nelder, J.~A. and Wedderburn, R. W.~M. (1972).
\newblock Generalized linear models.
\newblock {\em Journal of the Royal Statistical Society. Series A (General)},
  135(3):370--384.

\bibitem[Nguyen, 2017]{Nguyen:2017}
Nguyen, H. (2017).
\newblock Near universal consistency of the maximum pseudolikelihood estimator
  for discrete models.
\newblock {\em Annals of Statistics}, 2(2):22--23.

\bibitem[Plascak et~al., 1993]{Plascak:1993}
Plascak, J., Moreira, J., et~al. (1993).
\newblock Mean field solution of the general spin blume-capel model.
\newblock {\em Physics Letters A}, 173(4-5):360--364.

\bibitem[Plischke and Bergersen, 2006]{Plischke:1994}
Plischke, M. and Bergersen, B. (2006).
\newblock {\em Equilibrium statistical physics}.
\newblock World Scientific Publishing Company, 3rd edition.

\bibitem[P{\"o}tscher and Leeb, 2009]{Potscher:2009}
P{\"o}tscher, B.~M. and Leeb, H. (2009).
\newblock On the distribution of penalized maximum likelihood estimators: The
  lasso, scad, and thresholding.
\newblock {\em Journal of Multivariate Analysis}, 100(9):2065--2082.

\bibitem[Ravikumar et~al., 2010]{Ravikumar:2010}
Ravikumar, P., Wainwright, M., and Lafferty, J. (2010).
\newblock High-dimensional ising model selection using $\ell_1$-regularized
  logistic regression.
\newblock {\em The Annals of Statistics}, 38(3):1287--1319.

\bibitem[Romano et~al., 2012]{Romano:2012}
Romano, J.~P., Shaikh, A.~M., et~al. (2012).
\newblock On the uniform asymptotic validity of subsampling and the bootstrap.
\newblock {\em The Annals of Statistics}, 40(6):2798--2822.

\bibitem[Sch\"{a}fer and Strimmer, 2005]{Schafer2005}
Sch\"{a}fer, J. and Strimmer, K. (2005).
\newblock A shrinkage approach to large-scale covariance matrix estimation and
  implications for functional genomics.
\newblock {\em Statistical Applications in Genetics and Molecular Biology},
  4(1):32.

\bibitem[Schmid and Hunter, 2023]{Schmid:2023}
Schmid, C.~S. and Hunter, D.~R. (2023).
\newblock Computing pseudolikelihood estimators for exponential-family random
  graph models.
\newblock {\em Journal of Data Science}, 21(2).

\bibitem[Suggala et~al., 2017]{Suggala:2017}
Suggala, A.~S., Yang, E., and Ravikumar, P. (2017).
\newblock Ordinal graphical models: A tale of two approaches.
\newblock In {\em International conference on machine learning}, pages
  3260--3269. PMLR.

\bibitem[Tibshirani, 1996]{Tibshirani:1996}
Tibshirani, R. (1996).
\newblock Regression shrinkage and selection via the lasso.
\newblock {\em Journal of the Royal Statistical Society. Series B
  (Methodological)}, 58(1):267--288.

\bibitem[Vaart, 1998]{Vandervaart98}
Vaart, A. v.~d. (1998).
\newblock {\em Asymptotic Statistics}.
\newblock New York: Cambridge University Press.

\bibitem[van Borkulo et~al., 2014]{Borkulo:2014}
van Borkulo, C.~D., Borsboom, D., Epskamp, S., Blanken, T.~F., Boschloo, L.,
  Schoevers, R.~A., and Waldorp, L.~J. (2014).
\newblock A new method for constructing networks from binary data.
\newblock {\em Scientific Reports}, 4:5918.

\bibitem[van~de Geer et~al., 2014]{Geer:2014}
van~de Geer, S., B{\"u}hlmann, P., Ritov, Y., and Dezeure, R. (2014).
\newblock On asymptotically optimal confidence regions and tests for
  high-dimensional models.
\newblock {\em The Annals of Statistics}, 42(3):1166--1202.

\bibitem[van~der Maas et~al., 2026]{Maas:2026}
van~der Maas, H., Borsboom, D., and Waldorp, L. (2026).
\newblock The statistical physics of psychological networks: Zero matters.
\newblock {\em Psychological Review}, ((accepted)).

\bibitem[Wainwright, 2009]{Wainwright:2009}
Wainwright, M.~J. (2009).
\newblock Sharp thresholds for high-dimensional and noisy sparsity recovery
  using-constrained quadratic programming (lasso).
\newblock {\em IEEE Transactions on Information Theory}, 55(5):2183--2202.

\bibitem[Wainwright and Jordan, 2008]{Wainwright:2008}
Wainwright, M.~J. and Jordan, M.~I. (2008).
\newblock Graphical models, exponential families, and variational inference.
\newblock {\em Foundations and Trends in Machine Learning}, 1(1-2):1--305.

\bibitem[Waldorp and Haslbeck, 2024]{Waldorp:2024}
Waldorp, L. and Haslbeck, J. (2024).
\newblock Network inference with the lasso.
\newblock {\em Multivariate Behavioral Research}, 12(in press):1--20.

\bibitem[Waldorp and Kossakowski, 2020]{Waldorp:2020}
Waldorp, L. and Kossakowski, J. (2020).
\newblock Mean field dynamics of stochastic cellular automata for random and
  small-world graphs.
\newblock {\em Journal of Mathematical Psychology}, 97:102380.

\bibitem[Waldorp et~al., 2019]{Waldorp:2019}
Waldorp, L., Marsman, M., and Maris, G. (2019).
\newblock Logistic regression and ising networks: prediction and estimation
  when violating lasso assumptions.
\newblock {\em Behaviormetrika}, 46(1):49--72.

\bibitem[White, 1981]{White81}
White, H. (1981).
\newblock Consequences and detection of misspecified and nonlinear regression
  models.
\newblock {\em Journal of the American Statistical Association},
  76(374):419--433.

\bibitem[Williams, 2020]{Williams:2020b}
Williams, D.~R. (2020).
\newblock Beyond lasso: A survey of nonconvex regularization in gaussian
  graphical models.

\bibitem[Yang et~al., 2012]{Yang:2012}
Yang, E., Allen, G., Liu, Z., and Ravikumar, P.~K. (2012).
\newblock Graphical models via generalized linear models.
\newblock In {\em Advances in Neural Information Processing Systems}, pages
  1358--1366.

\bibitem[Young and Smith, 2005]{Young:2005}
Young, G. and Smith, R. (2005).
\newblock {\em Essentials of statistical inference}.
\newblock Cambridge University Press.

\end{thebibliography}

\appendix
\section*{Appendix}

\section{Mean field theory}\label{app:mean-field}\noindent
In the mean field we assume that each node is influenced in the same way by the average magnetisation $\mu$ and by on average $d$ other nodes with an average effect $\mu d \sigma$ \citep{Barrat:2008}. Mean field theory then results in taking the mean
\begin{align*}
\mu &= \sum_{x_s\in \{-1,0,+1\}} x_s \frac{ \exp\left(  \beta(\tau_s x_s + \mu d \sigma x_s  - \alpha^2 x_s^2) \right) }  {  \sum_{x_s\in \{-1,0,+1\}} \exp\left(  \beta(\tau_s x_s + \mu d\sigma x_s  - \alpha^2 x_s^2) \right) }\\
\end{align*}
Writing this out for variable $x_s$ equals $-1$, $0$ and $1$ gives
\begin{align*}
\mu &=
	\frac{ - \exp\left(  \beta(-\tau_s - \mu d\sigma   - \alpha^2 ) \right)  + \exp\left(  \beta(\tau_s + \mu d \sigma  - \alpha^2 ) \right)  }
		{ 1 + \exp\left(  \beta(-\tau_s - \mu d\sigma   - \alpha^2)  \right)  + \exp\left(  \beta(\tau_s + \mu d\sigma   - \alpha^2  )\right) }
\end{align*}
And this can be rewritten as
\begin{align} 
\mu 	&= \frac{2\sinh( \beta(\tau_s + \mu d\sigma))\exp(-\beta\alpha^2)  }
		{ 1 + 2\cosh( \beta(\tau_s + \mu d\sigma))\exp( -\beta\alpha^2 ) }
\end{align}
The self-consistent equation can be interpreted as the difference in proportion between the $-1$ and $1$ variables. This equation can also be obtained by minimising  the Gibbs free energy, given by \citep[][Section 4.3]{Plischke:1994}
\begin{align}\label{eq:free-energy}
G(\mu) = -\frac{1}{\beta}\log Z_\mu = 
\frac{1}{2}d\sigma\mu^2-\frac{1}{\beta}\log (1+ 2\exp(-\beta\alpha^2)\cosh(\beta \tau + \beta d \sigma \mu))
\end{align}

Similarly, the mean field for $x_s^2$ is obtained with the above average but then with $x_s^2$ instead of $x_s$. Then we obtain  \citep[see also,][equation (5)]{Kaufman:1990}
\begin{align*} 
\psi 	&= \frac{2\cosh( \beta(\tau_s + \mu d\sigma))\exp(-\beta\alpha^2)  }
		{ 1 + 2\cosh( \beta(\tau_s + \mu d\sigma))\exp( -\beta\alpha^2 ) }
\end{align*}
Note that $\psi$ is the proportion of nonzero (i.e., either $-1$ or $1$) variables, and $1-\psi$ is the proportion of 0s among the $m$ variables.

\section{Exponential family}\label{app:exponential-family}\noindent
We show in \ref{app:minimal-sufficient-statistic} that the BC model is exponential family and that the parameters are identifiable.  Furthermore, in \ref{app:properties-minimal-exponential-family} we discuss some properties of the exponential family that are either convenient to estimation, like the relation with maximum likelihood, or inconvenient, like the normalisation constant that is difficult to compute. 

\subsection{Minimal sufficient statistic of the BC model}\label{app:minimal-sufficient-statistic}\noindent
We assume throughout  that the normalising constant $Z_\theta$ (see equation (\ref{eq:normalising-constant})) is finite. To see that the BC model is exponential family, we write the distribution in canonical form \citep{Barndorff-Nielsen:1994}. Define the canonical parameter $\theta$ containing the vectors $(\tau_s, s\in V)$ and $(\sigma_{st}, (s,t)\in E)$ and $\alpha$, which has dimension $m=|V|+|E|+1$, where $|V|$ is the number of nodes $m$ and $|E|$ is the number of edges ($=m(m-1)/2$); in total we have $m(m+1)/2+1$ parameters. We further define the vector of canonical sufficient statistics $\phi(x)$ as the collection of $(x_s, s\in V)$, $(x_sx_t, (s,t)\in E)$ and $x_+$, where $x_+=\sum_{i \in V} x_i^2$. Then we can write the BC distribution as
\begin{align}\label{eq:bc-exponential-family}
p_\theta(x) = \exp\left( \theta^\top\phi(x) - \log Z_\theta \right)
\end{align}
where $Z_\theta$ is defined in (\ref{eq:normalising-constant}). This shows that the distribution is exponential family. 

To see that the parameter $\theta$ is identifiable, requires that $\phi(x)$ is a minimal sufficient statistic \citep[][Theorem 3.2]{Wainwright:2008}. That $\phi(x)$ is a sufficient statistic follows from the factorisation, i.e., the probability distribution is a product of a potential with the sufficient statistic as the only part including the parameter \citep[see e.g.,][Theorem 5.2]{Lehmann:1983}, which is immediate from equation (\ref{eq:bc-exponential-family}). To see that the sufficient statistic $\phi(x)$ is also minimal, requires that the vector of sufficient statistics $\phi(x)$ is non-redundant, i.e., there is no non-zero $a$ such that $a^\top \phi(x)$ equals a constant almost everywhere \citep{Barndorff-Nielsen:1994}.  This implies that for all $\mu=\mathbb{E}(\phi(X))$ we have that $a^\top \mu=b$ for some $b$. 
This is connected to the following: $\phi$ is minimal is equivalent to: for all $x\ne y \in \{-1,0,+1\}$ the ratio $\log(p_\theta(x)/p_\theta(y))=0$ if and only if $\phi(x)=\phi(y)$ \citep[][Corrollary 6.16]{Lehmann:1983} \citep[][Theorem 6.1]{Young:2005}. From equation (\ref{eq:bc-exponential-family}) we see that $\log(p_\theta(x)/p_\theta(y))=0$ is equivalent to $\theta^\top (\phi(x)-\phi(y))=0$. And the only non-trivial case is if $\phi(x)=\phi(y)$, but this implies that for each $s\in V$, $x_s=y_s$, because $x_s=\phi(x_s)=\phi(y_s)=y_s$; a contradiction. Hence, $\phi$ is a minimal sufficient statistic. Note that this argument can be extended to the model where we have $\alpha_s^2$ for all $s\in V$. 
\subsection{Properties of minimal exponential family distribution}\label{app:properties-minimal-exponential-family}\noindent
The log partition function $\log Z_\theta$ (see (\ref{eq:normalising-constant}) for the definition of $Z_\theta$) is also referred to as the cumulant function and serves to obtain moments of the distribution of $\phi_s(X)$ \citep[][]{McCulloch:1988}. To derive the first order moments of the sufficient statistics in $\phi(x)$ we can use the relation \citep[][Proposition 3.1]{Wainwright:2008}
\begin{align*}
\nabla_s\log Z_\theta=\frac{\partial \log Z_\theta}{\partial \theta_s} = \mathbb{E}(\phi_s(X)) 
\end{align*}
Applying this to the $\log$ of equation (\ref{eq:normalising-constant}) with $\theta_s=\tau_s$, so that $\phi_s(x)=x_s$, gives 
\begin{align*}
  \frac{\partial}{\partial \tau_s}  \log Z_\theta=
  \frac{\partial}{\partial \tau_s}  \log
\sum_{x\in \{-1,0,+1\}^V} \exp\left( \sum_{i\in V} \tau_i x_i + \sum_{(i,j)\in E} \sigma_{ij} x_i x_j + \alpha^2 \sum_{i\in V} x_i^2\right)
\end{align*}
It is easily seen that this leads to the expectation
\begin{align}\label{eq:mean-function-log-partition}
\mathbb{E}(X_s)=\sum_{x\in\{-1,0,+1\}^V} x_s 
\frac{1}{Z_\theta}\exp\left( \sum_{i\in V} \tau_i x_i + \sum_{(i,j)\in E} \sigma_{ij} x_i x_j + \alpha^2 \sum_{i\in V} x_i^2 \right)
\end{align}
which is seen to be the same as $\mathbb{E}(X_s)=\sum_{x\in\{-1,0,+1\}^V} x_s p_\theta(x)$. In the same way, we obtain for $\sigma_{st}$ the first-order derivative
\begin{align*}
\nabla_t\log Z_\theta =\mathbb{E}(X_s X_t)= \sum_{x\in\{-1,0,+1\}^V} x_s x_t p_\theta(x)
\end{align*}
and the first-order derivative for $\alpha^2$
\begin{align*}
\nabla_s\log Z_\theta =\mathbb{E}(X_s^2)= \sum_{x\in\{-1,0,+1\}^V} x_s^2 p_\theta(x)
\end{align*}

Also from the map $\nabla \log Z_\theta$ we obtain the covariances of the sufficient statistics $\text{cov}(\phi_s(X),\phi_t(X))$. This is obtained by taking the second-order derivatives $\nabla^2 \log Z_\theta$ \citep[][Theorem 3.1]{Wainwright:2008}. So, we obtain 
\begin{align*}
\nabla^2_{st} \log Z_\theta = \frac{\partial^2}{\partial \theta_s\partial\theta_t} \log Z_\theta = \text{cov}(\phi_s(X),\phi_t(X))
\end{align*}
Applying this to the first-order derivative in equation (\ref{eq:mean-function-log-partition}) with $\theta_t=\tau_t$ gives 
\begin{align*}
\mathbb{E}(X_sX_t)=\sum_{x\in\{-1,0,+1\}^V} x_sx_t p_\theta(x) - \sum_{x\in\{-1,0,+1\}^V} x_s p_{\theta}(x)\sum_{x\in\{-1,0,+1\}^V} x_t p_{\theta}(x)
\end{align*}
The second-order derivatives for $\sigma_{st}$ and $\sigma_{uv}$, corresponding to $\phi_k(X)=X_s X_t$ and $\phi_j(X)=X_u X_v$, obtain fourth-order moments $\phi_k(X)\phi_j(X)=X_sX_tX_uX_v$. 

We see that to evaluate the means and covariances of the sufficient statistics through the cumulant function requires computation of the partition function $Z_\theta$ in equation (\ref{eq:normalising-constant}). This is not feasible and so other alternatives like pseudo-likelihood are often used (see Appendix \ref{app:pseudo-likelihood}). 

The pre-image of $\nabla \log Z_\theta$, which is also the negative entropy, obtains maximum likelihood estimates whenever the exponential family distribution is minimal \citep[][Chapter 3]{Wainwright:2008}, as is the case for the BC model. We can obtain maximum likelihood estimates by maximising the quantity
\begin{align}
\ell_\theta^n(X) = \frac{1}{n}\sum_{i=1}^n \log p_\theta(X_i) = \theta^\top \hat{\mu} - \log Z_\theta
\end{align}
where $\hat{\mu}$ are the estimates $\tfrac{1}{n}\sum_{i=1}^n \phi(X_i)$. The maximum is obtained by taking the first-order derivative with respect to $\theta$, setting it to 0 and solving for $\theta$ \citep{Vandervaart98}. The pre-image of $\nabla \log Z_\theta$ now tells us that because we have a bounded mean space $\mathbb{M}$ (see Figure \ref{fig:marginal-polytope}), this implies that the parameter space $\Omega=\mathbb{R}^d$ \citep[see, e.g.,][Corollary 6.16, Chapter 1]{Lehmann:1983}, with parameter dimension $d=|V|+|E|+1$ for the BC model.

\section{Pseudo-likelihood}\label{app:pseudo-likelihood}\noindent
Here we consider a single pseudo-likelihood and ignore the fact that we use multiple pseudo-likelihoods to `stitch' together the entire BC model. Often the indices with respect to a single pseudo-likelihood are therefore ignored, e.g., we use $\tau$ instead of $\tau_s$ etc. 
\subsection{Consistency of pseudo-likelihood}\label{app:consistency-pseudo-likelihood}\noindent
For the BC model we have the log pseudo-likelihood function for a single observation of the random variable  $X=(X_1,X_2,\ldots,X_m)$ as the observed value $x$
\begin{align}
\ell_\theta(x) =
\tau x_{s} + x_{s}\sum_{t\ne s} \sigma_{st} x_{t} - \alpha^2 x_{s}^2	
  - \log \left(1 + 2\cosh\left(\tau + \sum_{t\ne s} \sigma_{st} x_{t}\right)\exp(-\alpha^2)\right)
\end{align}
The sum over multiple observations obtained from the same BC model are assumed in equation (\ref{eq:pseudo-likelihood-estimation}) to obtain the pseudo-likelihood estimate $\hat{\theta}$. 

The convergence of the estimate $\hat{\theta}$ is in terms of the Euclidean distance $||\hat{\theta}-\theta||_2^2$ which goes to 0 as $n\to \infty$. This follows from the proof in \citet[][Remark 2]{Nguyen:2017}. The proof has relatively mild assumptions. \\
\begin{itemize}
\item[1.] (Identifiability) For each mean $\mu\in \mathbb{M}$ there is a unique $\theta\in \Omega$. 
\item[2.] (True model) For any $n$ the estimate $\hat{\theta}$ and true parameter $\theta$ are in the set $\Omega$.\\
\end{itemize}
The second assumption can be relaxed by assuming that the best approximate model in the sense of Kullbeck-Leibler divergence is in $\Omega$ \citep{McCulloch:1988}. Note that the first assumption holds whenever $\mu$ is in the interior of $\mathbb{M}$ (see Figure \ref{fig:marginal-polytope}). And this is a consequence of identifiability \citep[][Theorem 3.3]{Wainwright:2008}. The BC model (joint distribution) was shown to be minimal in \ref{app:minimal-sufficient-statistic}, and hence, the mapping $\theta\mapsto p_\theta(x))$ is injective and $\theta$ is identifiable. It is easily seen that a similar argument holds for the pseudo-likelihood, such that the mapping $\theta\mapsto p_\theta(x_s\mid x_{t\ne s})$ is also injective. 

The pseudo-likelihood in equation (\ref{eq:pseudo-likelihood}) is for the BC model
\begin{align}\label{eq:pseudo-likelihood-bc}
p_\theta(x_1\mid x_{t\ne 1}) & \cdots  p_\theta(x_m\mid x_{t\ne m}) \notag \\
&= 
\frac{\exp \left( \sum_{s\in V}\tau_s x_s + 2\sum_{(s,t)\in E} \sigma_{st}x_s x_t -\alpha^2 \sum_{s\in V}x_s^2\right) }
{\prod_{s\in V}(1 + 2\cosh(\tau_s + \sum_{t\ne s} \sigma_{st} x_t)\exp(-\alpha^2)) }
\end{align}
We observe that we obtain an approximation to the joint distribution but with a different normalisation constant. For identifiability, the same argument as used in Appendix \ref{app:minimal-sufficient-statistic} holds.

From the conditional distributions $p_\theta(x_j\mid x_{i\neq j})$ in (\ref{eq:conditional-distribution}) used in the pseudo-likelihood, we can obtain the mean and variances of the variables $X_s$ and $X_v$. It is easy to derive that for the variable $X_s$
\begin{align}\label{eq:conditional-expectation}
\mathbb{E}(X_s\mid x_{t\ne s}) 
			= \nabla_\tau\log Z_\theta
			= \frac{2\sinh(\tau_s + \sum_{t\ne s} \sigma_{st} x_t)\exp(-\alpha^2)} { 1 + 2\cosh(\tau_s + \sum_{t\ne s} \sigma_{st} x_t)\exp(-\alpha^2) }
\end{align}
And for the variables $X_s$ and $X_v$, taking the derivative of $\log Z_\theta$ with respect to $\sigma_{st}$ we obtain
\begin{align}
\mathbb{E}(X_sX_v\mid x_{t\ne s}) 
			= \nabla_{\sigma}\log Z_\theta
			= \frac{2\sinh(\tau_s + \sum_{t\ne s} \sigma_{st} x_t)\exp(-\alpha^2)} { 1 + 2\cosh(\tau_s + \sum_{t\ne s} \sigma_{st} x_t)\exp(-\alpha^2) }x_v
\end{align}
Finally, for the squares $X_s^2$, connected to $\alpha^2$, we obtain the expectation 
\begin{align}
\mathbb{E}(X_s^2\mid x_{t\ne s}) 
			= \nabla_{\alpha^2}\log Z_\theta
			= \frac{2\cosh(\tau_s + \sum_{t\ne s} \sigma_{st} x_t)\exp(-\alpha^2)} { 1 + 2\cosh(\tau_s + \sum_{t\ne s} \sigma_{st} x_t)\exp(-\alpha^2) }
\end{align}

For comparison, the mean field obtains a similar expectation as the conditional expectation in equation (\ref{eq:conditional-expectation}). Note that from (\ref{eq:conditional-expectation}), with $\sigma_{st}=1$ if the edge $(s,t)$ is present and 0 otherwise, and $x_t$ is replaced by the average effect $\sigma\mu_s$, then we get $\sum_{s\ne t}\sigma_{st}x_t=\mu\sigma d$, and the mean field equals the conditional probability.

\subsection{Lasso estimation and assumptions}\label{app:lasso-estimation}\noindent
For accurate estimation with the lasso several assumptions are required \citep{Wainwright:2009,Buhlmann:2014b}. We use the convergence result compatible with the convergence result above for the pseudo-likelihood (without the lasso), taken from \citet{Geer:2014}. To introduce the assumptions we use the notation $\theta_{S_0}$, where the set $S_0\subseteq\{1,2,\ldots,m\}$ contains indices for the true non-zero parameters of the BC model to indicate that $\theta_{S_0,k}=\theta_k$ if $k\in S_0$ and $\theta_{S_0,k}=0$ if $k\notin S_0$. \\
\begin{itemize}

\item[(L1)] The sparsity $s_0$ (number of non-zero parameters) is of the order $\sqrt{n}/\log m^\star$, where $m^\star=m(m+1)/2$ and $m$ is the number of nodes. 

\item[(L2)] The compatibility condition holds such that for all $\theta$ satisfying $||\theta_{S_0^c}||_1\le 3 ||\theta_{S_0}||_1$,
\begin{align*}
||\theta_{S_0}||_1\le s_0\theta^\top (\nabla^2\ell^n_\theta) \theta/\phi^2
\end{align*}
with compatibility constant $\phi^2>0$. \\
\end{itemize}
With these two assumptions it can be shown \citep[][]{Geer:2014} that if the penalty $\lambda$ for the lasso is of order $\sqrt{\smash[b]{\log m^\star/n}}$, with $m^\star=m(m+1)/2$. then $||\hat{\theta}-\theta_0||_1$ converges to 0.

\subsection{Standard errors}\label{app:standard-errors}\noindent
We want to obtain $95\%$ confidence intervals for parameter $\hat{\theta}_k$ with standard error $\Gamma_{kk}/\sqrt{\smash[b]{n}}$ of the form 
\begin{align*}
(\hat{\theta}^d_{k} - 1.96\sqrt{\Gamma_{kk}/n}, \hat{\theta}^d_{k} + 1.96\sqrt{\Gamma_{kk}/n})
\end{align*}
where $\hat{\theta}_{k}^d$ is the desparsified lasso estimate of one of the BC model parameters, and $\Gamma$ is the parameter variance matrix, evaluated at the lasso estimate. We require accurate estimates of $\Gamma$ for this. 

To obtain standard errors for the pseudo-likelihood we use the sandwich estimate \citep{White81}. This is a solution for both the fact that the pseudo-likelihood is used (and hence there is model misspecification) and that the desparsified lasso is applied to obtain accurate confidence intervals for lasso type estimators. 

%

The desparsified lasso is the lasso obtained by minimising equation (\ref{eq:lasso}) and subtracting (desparsifying) a projection of the residual \citep[][Section 3.1]{Geer:2014}
\begin{align}\label{eq:desparsified-lasso}
\hat{\theta}_d = \hat{\theta} - (\nabla^2\ell_\theta)^{-1}\nabla \ell_\theta
\end{align}
where $\hat{\theta}$ is the lasso estimate and $\ell_\theta$ is the log pseudo-likelihood function without regularisation. Because the first-order derivative is approximately normal with mean 0, the variance is the sandwich estimate \citep{Javanmard:2014}
\begin{align}\label{eq:sandwich}
\Gamma= (\nabla^2\ell_\theta)^{-1}[\nabla \ell_\theta \nabla \ell_\theta^\top](\nabla^2\ell_\theta)^{-\top}
\end{align}
where $\nabla\ell_\theta$ and $\nabla^2\ell_\theta$ are evaluated at the lasso estimate, and $A^{-\top}$ means that both the inverse and the transpose is taken. It is called the sandwich estimator, since the the first-order derivatives are "sandwiched" in between the inverse of the second-order derivatives. The distribution of the desparsified lasso is continuous and approximately Gaussian \citep{Geer:2014}.

The motivation for the sandwich estimate from model misspecification (here for using the pseudo-likelihood) is obtained by using an approximation of the first-order derivative based on the pseudo-likelihood estimate $\hat{\theta}$ and the true parameter $\theta$  \citep[or best parameter in terms of the smallest Kullbeck-Leibler divergence][Chapter 9]{Claeskens:2008,Young:2005}. This approximation is identical to equation (\ref{eq:desparsified-lasso}), except that the estimate $\hat\theta$ need not be the lasso in general. 
%
%
See \citet[][Chapter 5]{Vandervaart98} for general statements about when such results hold, and \citet[][Section 3]{Geer:2014} for specific statements regarding the lasso and generalised linear models. 

For the BC model, we can specify the first- and second-order derivatives further to obtain the sandwich estimator.  
We obtain the $p+1$ ($=p-1+2$ for the pseudo-likelihood) vector of first- and second-order derivatives with respect to the parameters $\tau$, $\sigma_{st}$ and $\alpha^2$.
 The first-order derivatives for the BC model are for the parameters $\tau$, $\sigma$ and $\alpha^2$ (where we omit the indices for convenience) is contained in the vector 
 \begin{align*}
 \nabla\ell_\theta = (\nabla_\tau\ell_\theta,\nabla\sigma\ell_\theta,\nabla_{\alpha}\ell_\theta)^\top
 \end{align*}
 The elements of this vector are
\begin{align*}
\nabla_\tau \ell_\theta(x) = \phi_\tau(x) - \nabla_\tau Z_\theta = x_s - \frac{ 2\sinh(\gamma)\exp(-\alpha^2)}{1+2\cosh(\gamma)\exp(-\alpha^2)}
\end{align*}
where $\gamma=\tau+\sum_{t\ne s}\sigma_{st}x_t$, the "Ising part" of the BC model. The first-order derivatives with respect to $\sigma_{st}$ is 
\begin{align*}
\nabla_{\sigma} \ell_\theta(x) = \phi_{\sigma}(x) - \nabla_{\sigma} Z_\theta = x_s x_t - \frac{ 2\sinh(\gamma)\exp(-\alpha^2)}{1+2\cosh(\gamma)\exp(-\alpha^2)}x_t
\end{align*}
And finally, the first-order derivatives with respect to $\alpha^2$ is 
\begin{align*}
\nabla_{\alpha^2} \ell_\theta(x) = \phi_{\alpha^2}(x) - \nabla_{\alpha^2} Z_\theta = -x_s^2 + \frac{ 2\cosh(\gamma)\exp(-\alpha^2)}{1+2\cosh(\gamma)\exp(-\alpha^2)}
\end{align*}

For the second-order derivatives we obtain the matrix 
\begin{align*}
\nabla^2 \ell_\theta
=
\begin{pmatrix}
\nabla^2_{\tau\tau}\ell_\theta	&\nabla^2_{\tau\sigma}\ell_\theta	&\nabla^2_{\tau\alpha}\ell_\theta\\
\nabla^2_{\sigma\tau}\ell_\theta	&\nabla^2_{\sigma\sigma}\ell_\theta	&\nabla^2_{\sigma\alpha}\ell_\theta\\
\nabla^2_{\alpha\tau}\ell_\theta	&\nabla^2_{\alpha\sigma}\ell_\theta	&\nabla^2_{\alpha\alpha}\ell_\theta
\end{pmatrix}
\end{align*}
For its elements we have for the derivatives with respect to $\tau$
\begin{align*}
\nabla_{\tau\tau}^2 \ell_\theta(x) =  -\nabla_{\tau\tau}^2 Z_\theta = -\frac{ 2\cosh(\gamma)\exp(-\alpha^2)+4\exp(-2\alpha^2)}{(1+2\cosh(\gamma)\exp(\alpha^2))^2}
\end{align*}
The second-order derivative for parameters $\sigma_{st}$ and $\sigma_{sv}$ is
\begin{align*}
\nabla_{\sigma\sigma}^2 \ell_\theta(x) =  -\nabla_{\sigma\sigma}^2 \log Z_\theta = -\frac{ 2\cosh(\gamma)\exp(-\alpha^2)+4\exp(-2\alpha^2)}{(1+2\cosh(\gamma)\exp(-\alpha^2))^2}x_t x_v
\end{align*}
And with respect to $\alpha^2$ is
\begin{align*}
\nabla_{\alpha^2\alpha^2}^2 \ell_\theta(x) =  -\nabla_{\alpha^2\alpha^2}^2 \log Z_\theta = \frac{ 2\cosh(\gamma)\exp(-\alpha^2)}{(1+2\cosh(\gamma)\exp(-\alpha^2))^2}
\end{align*}
The three second-order derivatives with respect to different parameters are
\begin{align*}
\nabla_{\tau \sigma}^2 \ell_\theta(x) =  -\nabla_{\tau \sigma}^2 \log Z_\theta =- \frac{ 2\cosh(\gamma)\exp(-\alpha^2)+4\exp(-2\alpha^2)}{(1+2\cosh(\gamma)\exp(\alpha^2))^2}x_t
\end{align*}
\begin{align*}
\nabla_{\tau \alpha^2}^2 \ell_\theta(x) =  -\nabla_{\tau \alpha^2}^2 \log Z_\theta = -\frac{ 2\sinh(\gamma)\exp(-\alpha^2)}{(1+2\cosh(\gamma)\exp(-\alpha^2))^2}
\end{align*}
\begin{align*}
\nabla_{\alpha^2 \sigma}^2 \ell_\theta(x) =  -\nabla_{\alpha^2\sigma}^2 \log Z_\theta = - \frac{ 2\sinh(\gamma)\exp(-\alpha^2)}{(1+2\cosh(\gamma)\exp(-\alpha^2))^2}x_t
\end{align*}
\section{Combining shrinkage and desparsified lasso}\label{app:dLshrink}\noindent
The objective of the shrinkage estimate in the current context is to obtain an approximation of the second-order derivatives $\nabla^2\ell_\theta$ in equation (\ref{eq:sandwich}), which then leads to accurate confidence intervals. The shrinkage estimator is a combination of the original estimate and a ''shrunken`` version, a diagonal matrix $\mu I$, where $I$ is the diagonal matrix, defined as \citep{Schafer2005}
\begin{align}
\Sigma = \rho \mu I + (1-\rho)\nabla^2\ell_\theta
\end{align}
where $\rho$ is a constant such that it is in $0\le \rho\le 1$; $\rho=0$ leads to the original estimate and $\rho=1$ leads to the estimate where all dependence is shrunk to 0. This shrinkage matrix $\Sigma$ is invertible for appropriate $\mu$, and this inverse we call $\Theta=\Sigma^{-1}$. Here we show that the shrinkage version $\Theta$ can be used such that we obtain the appropriate weak convergence and subsequent confidence intervals.  Let $||A||$ be the Frobenius norm $\sqrt{\smash[b]{\text{tr}(A^\top A)}}$. We require several assumptions for this to work, which are similar to the ones presented in \citep[][C1 to C8]{Geer:2014}. \\
\begin{itemize}

\item[(A1)] The first- and second-order partial derivatives exist and are bounded. Also the second order derivatives are Lipschitz continuous of order 1 in the parameters, i.e., for each single pair of parameters $\theta$ and $\psi$ in $\Omega$
\begin{align*}
\max_{\theta,\psi}\frac{|\nabla^2\ell_\theta(x) - \nabla^2\ell_\psi(x)|}{|\theta-\psi|}\le 1
\end{align*}

\item[(A2)] We use the assumptions (L1) and (L2) of the lasso to ensure that $||\hat{\theta}-\theta_0||_1\to 0$: (L1) The sparsity $s$ is such that $s=o(\sqrt{\smash[b]{n}}/\log p^\star)$; (L2) the compatibility assumption holds with $\phi^2>0$ for all $\theta$ such that $||\theta_{S_0^c}||_1\le 3||\theta_{S_0}||_1$.

\item[(A3)] The eigenvalues $\delta_t$ of $\nabla^2\ell_\theta$ are bounded, i.e., $\delta_t=O_p(1)$.

\item[(A4)] The shrinkage estimate of the common variance $\mu$ is bounded, i.e., $\mu=O_p(1)$

\item[(A5)] Given the shrinkage estimate $\Theta$ of the inverse of the  second-order derivative $\nabla^2\ell_\theta$, we have that
\begin{align*}
||\Theta\nabla^2\ell_\theta - I||=O(\lambda)
\end{align*}
This assumption is essential to obtain correct coverage of the confidence intervals. 

\item[(A6)] For each $k$, $1/\Gamma_{kk}$ is bounded and the random variable 
\begin{align*}
\frac{(\nabla^2\ell_\theta^{-1}\nabla\ell_\theta)_k}{\Gamma_{kk}/\sqrt{n}}
\end{align*}
is normally distributed with mean 0 and variance 1. \\
\end{itemize}
Assuming (A1) to (A6) means we can apply \citet[Theorem 3.1 from ][]{Geer:2014} so that we can obtain the confidence intervals. We now show that using the shrinkage we obtain Assumption (A5). 
\begin{lemma}
Assume (A1) to (A4) and $\lambda=O(\sqrt{\smash[b]{\log m^\star/n}})$ and let $\Theta=\Sigma^{-1}$ be the shrinkage estimate for the inverse of $\nabla^2\ell_\theta$. Then, for some $\gamma>0$ and $\rho=1/n^{1+\gamma}$ we obtain that 
\begin{align*}
||\Theta\nabla^2\ell_\theta - I||=O(\lambda)
\end{align*}
\end{lemma}
\begin{proof}
We will use the eigenvalue decomposition of $\nabla^2\ell_\theta$ and the shrinkage estimate $\Sigma$ to obtain a sufficient condition for the shrinkage parameter $\rho$, so that we get the bound in the lemma. 

We have that the shrinkage estimator $\Sigma$ of $\nabla^2\ell_\theta$ can be written as $\Sigma=U(\rho \mu I + (1-\rho) D)U'$, with $D$ a diagonal matrix with eigenvalues $\delta_t$. Its inverse $\Sigma^{-1} = \Theta$ can be written as 
\begin{align*}
\Theta= \sum_{t\in V}u_{t}u_{t}^{\sf T}\frac{\delta_{t}}{\rho\mu +(1-\rho)\delta_{t}}
\end{align*}
The Frobenius norm is invariant to orthonormal transformations, and so we can consider
\begin{align*}
||\Theta\nabla^2\ell_\theta-I ||
	\le \max_{t\in V}\left| \frac{\delta_{t}}{\rho\mu + (1-\rho)\delta_{t}}-1\right|
\end{align*}
To have the left hand side in the above equation correspond to the order of $\lambda$ in the right hand side, it is sufficient to consider
\begin{align*}
\max_{t\in V}\left| \frac{\delta_{t}}{\rho\mu + (1-\rho)\delta_{t}}-1\right| = O_p\left(\sqrt{\frac{\log p^\star}{n}}\right)
\end{align*}
We assumed  in (A3) that the eigenvalues $\delta_s=O_p(1)$. This implies that for $t$ such that the difference above is maximal
\begin{align*}
\frac{\sqrt{n}\delta_{t}}{\sqrt{\log p^\star}\rho\mu + \sqrt{\log p^\star}(1-\rho)\delta_{t}}-\sqrt{\frac{n}{\log p^\star}}
\end{align*}
such that this difference remains bounded for any $n$ and $p$. A sufficient condition for this to happen is that both terms are of the same order. It is easily seen that, because $\mu=O_p(1)$ by assumption (A4), if we take $\rho=o(1/n)$ then $\sqrt{\smash[b]{\log p^\star}}\rho\mu$ converges to 0 (because $\sqrt{\smash[b]{\log p^\star}/n}$ is bounded), and so the difference above converges to 0 and $\sqrt{\smash[b]{\log p^\star}}(1-\rho)\delta_t$ converges to $\sqrt{\smash[b]{\log p^\star}}\delta_t$. Hence, we can choose $\rho=1/n^{1+\gamma}$, for some $\gamma>0$ such that
\begin{align*}
\frac{\sqrt{n}\delta_{t}}{\sqrt{\log p^\star}\rho\mu + \sqrt{\log p^\star}(1-\rho)\delta_{t}} \to \lambda
\end{align*}
This completes the proof. 
\end{proof}
%

\section{Gibbs sampler for the BC model}\label{app:gibbs-sampler}\noindent
Here we provide pseudo-code for the algorithm for the Gibbs sampler of the BC model. As explained in, e.g., \citet[][]{Haggstrom:2002}, each node is selected randomly with equal probability $1/|V|$, and then assignment to one of $\{-1,0,+1\}$ is based on the probability for each value and a random variable $U$ from the uniform distribution on $[0,1]$. This is repeated $niter$ times, which is typically large, like 10,000.
\begin{algorithm}
\caption{Gibbs sampler BC model}\label{alg:gibbs-sampler}
\begin{algorithmic}
\State $i \gets 1$
\For{$i\le niter$} 
  \For{$v$ in the set of nodes $V$}
    \State sample $v\in V$ with equal probability $\frac{1}{|V|}$
    \State $p_\theta(X_v=-1\mid x_{j\ne v}) =  \frac{\exp ( -\beta(-\tau_s -\sum_{j\ne s} \sigma_{sj} x_j + \alpha^2 x_s^2)) }
					{ 1 + 2\cosh(-\beta(\tau_i + \sum_{j\ne s} \sigma_{sj} x_j))\exp(-\beta\alpha^2) }$\\
    \State $p_\theta(X_v=0\mid x_{j\ne v}) =  \frac{1 }
					{ 1 + 2\cosh(-\beta(\tau_i + \sum_{j\ne s} \sigma_{sj} x_j))\exp(-\beta\alpha^2) }$\\
    \State  obtain $u$ uniformly from $[0,1]$
    \If{$p_\theta(X_v=-1\mid x_j\ne v)\leq u$} $X_v \leftarrow -1$
    \ElsIf{$p_\theta(X_v=-1\mid x_{j\ne v})\ge u$ and $p_\theta(X_v=0\mid x_{j\ne v})\leq u$\\ \qquad\qquad}{ $X_v \leftarrow 0$}
    \ElsIf{$p_\theta(X_v=0\mid x_{j\ne v})\ge u$} $X_v \leftarrow 1$
     \EndIf
  \EndFor
  \State $i \leftarrow i+1$
\EndFor
\end{algorithmic}
\end{algorithm}
%

\section{Details about the numerical illustration}\label{app:details-numerical-illustrations}\noindent
For the simulation we generated a random matrix with $p_e=0.3$, so that we obtain on average a network with $13.5$ edges for the 10 node network, $57$ edges for the 20 node network, and approximately $131$ edges for the $30$ node network. We used the Gibbs sampler from Appendix \ref{app:gibbs-sampler} to generate data for $n=50$ to $260$ with intervals of $30$. The parameters for the Gibbs sampler were $\beta=2$, $\alpha^2=1.2$, $\tau=1.2$ and $\sigma_{st}=1$ for all $s,t\in V$. For each run in any sample size an Erd\"{o}s-Renyi network was generated with $p_e=0.3$. The Gibbs sampler was run for 2000 iterations to obtain a reasonable approximation to the Blume-Capel distribution.  

To estimate the parameters we implemented the pseudo-likelihood algorithm in \texttt{R} using the optimisation function \texttt{optim}. The default optimisation algorithm, also used for the simulations, was the option \texttt{method="BFGS"}. This refers to a quasi-Newton optimisation algorithm that uses the gradient to determine the direction and step size to evaluate the pseudo-likelihood surface. For estimation we set $\lambda = \sqrt{\log(p)/n}$ for the network parameters only (it is possible to regularise the thresholds and zero-parameters also). The R-code to apply this to data with three states is available at XXX. 

Standard errors were obtained with the shrinkage estimator \citep{Schafer2005} with shrinkage parameter $\rho=1/n^{5/4}$ (see Appendix \ref{app:dLshrink}). This seems to work well for good coverage of confidence intervals and reasonably sized confidence intervals. \citet{Schmid:2023} suggest to obtain the variance of the first order derivatives $\mathbb{E}[\nabla \ell_\theta(x)\nabla \ell_\theta(x)^\top]$ by simulating from the parameters obtained by pseudo-likelihood estimation and then use the variation from these simulations. Here, we use the diagonal elements of the first-order derivatives, which works well as shown in the simulations.

\section{Additional analyses application to voting data}\label{app:application-additional}\noindent
In addition to analysing the entire dataset with $n=10000$ observations, we also performed subsampling analyses \citep{Romano:2012}. In each of the 50 subsample analyses, we randomly selected $0.70$ of the entire sample, so that each subsample contains $n_s=0.70n=7000$ observations. We repeated the analysis with the exact same settings, with the lasso penalty adjusted to the sample size, so that $\lambda=\sqrt{\smash[b]{\log(m)/n_s}}=0.052$ and in the standard error estimates using the shrinkage estimate we adjusted the sample size for $1/n_s^{5/4}$. The result of the subsample anaylsis is shown in Figure \ref{fig:stem-subsample}(a) for the network and confidence intervals for a randomly chosen subsample in Figure \ref{fig:stem-subsample}(b).
\begin{figure}\centering
\begin{tabular}{c @{\hspace{0em}} c }
\raisebox{1em}{\pgfimage[width=.48\textwidth]{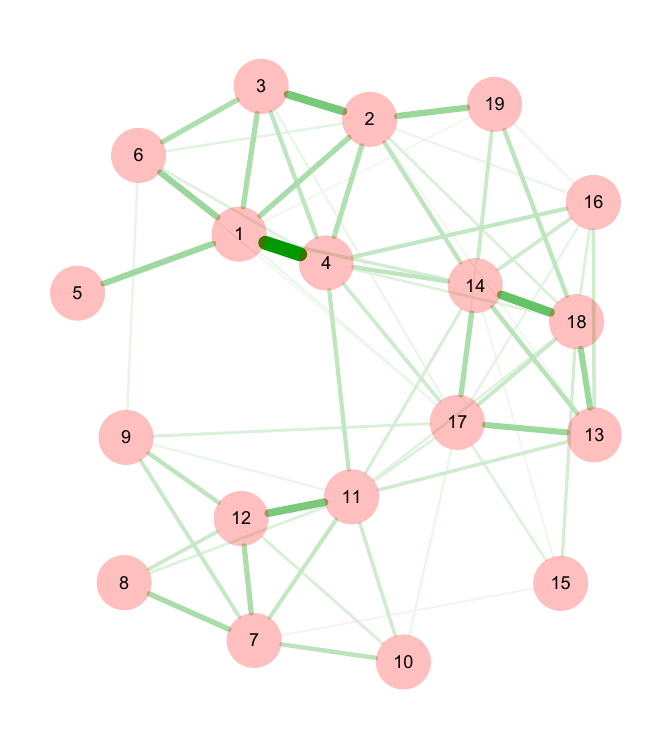}}		&\pgfimage[width=.52\textwidth]{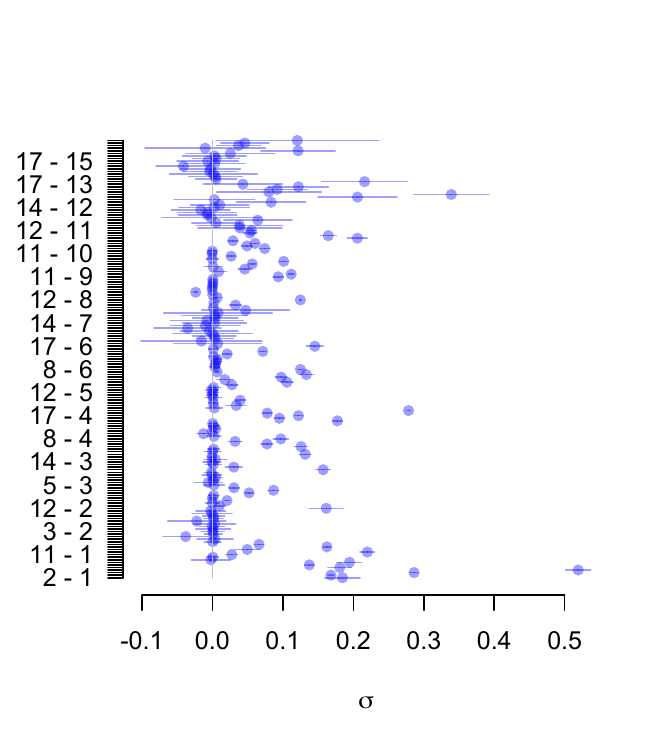}\\
	(a)		&(b)
\end{tabular}
\caption{In (a) is the network estimate, averaged across 50 subsamples, using the BC model with the lasso algorithm, thresholded using $\sqrt{\smash[b]{\log(p)/n_s}}$ with $n_s=0.70n$; edges below this threshold are set to 0. In (b) are the confidence intervals of the edge parameters $\sigma$ based on one of the subsamples, using the shrinkage estimate described in Appendix \ref{app:dLshrink}.  }
\label{fig:stem-subsample}
\end{figure}
%

\end{document}